\newif\ifSTOC
\STOCfalse

\ifSTOC
\documentclass{sig-alternate}

\else
\documentclass{article}
\usepackage{fullpage}
\fi

\usepackage{url}
\usepackage{amsmath}
\usepackage{amssymb}
\usepackage{graphicx}
\usepackage{stmaryrd}

\usepackage[colorlinks,urlcolor=blue,citecolor=blue,linkcolor=blue]{hyperref}
\newcommand{\normF}[1]{{\| #1 \|}_F}
\newcommand{\norm}[1]{{\| #1 \|}}

\newcommand{\cE}{\mathcal{E}}
\DeclareMathOperator{\nnz}{\mathtt{nnz}}
\DeclareMathOperator{\tr}{\mathtt{tr}} 
\DeclareMathOperator{\rank}{\mathtt{rank}}
\DeclareMathOperator{\argmin}{\mathrm{argmin}}
\DeclareMathOperator{\E}{\mathbf{E}}
\DeclareMathOperator{\Var}{\mathbf{Var}}
\newcommand{\ZZ}{S}

\DeclareMathOperator{\cond}{\mathrm{cond}}
\newcommand{\Ibr}[1]{\llbracket#1\rrbracket}
\newcommand\scn{_{s:n}}
\newcommand\vk{1/a}
\newcommand\kv{a}
\newcommand\fvk{\frac{1}{a}}
\newcommand\gdelta{\delta}
\newcommand\tO{\tilde{O}}

\newcommand\STOComitedproof[1] {%
\ifSTOC%
\begin{proof} Omitted in this version\end{proof}%
\else #1 \fi%
}

\newcommand{\poly}{{\mathrm{poly}}}
\newcommand{\eps}{\varepsilon}
\newcommand{\R}{{\mathbb R}}
\newcommand{\polylog}{{\mathrm{polylog}}}

\newtheorem{theorem}{Theorem}

\newtheorem{lemma}[theorem]{Lemma}

\newtheorem{corollary}[theorem]{Corollary}

\newtheorem{fact}[theorem]{Fact}

\newtheorem{remk}[theorem]{Remark}
\newtheorem{exmp}[theorem]{Example}

\def\FullBox{\hbox{\vrule width 8pt height 8pt depth 0pt}}

\def\qed{\ifmmode\qquad\FullBox\else{\unskip\nobreak\hfil
\penalty50\hskip1em\null\nobreak\hfil\FullBox
\parfillskip=0pt\finalhyphendemerits=0\endgraf}\fi}

\def\qedsketch{\ifmmode\Box\else{\unskip\nobreak\hfil
\penalty50\hskip1em\null\nobreak\hfil$\Box$
\parfillskip=0pt\finalhyphendemerits=0\endgraf}\fi}

\ifSTOC\else
\newenvironment{proof}{\begin{trivlist} \item {\bf Proof:~~}}
  {\qed\end{trivlist}}
\fi

\newenvironment{proofof}[1]{\begin{trivlist} \item {\bf Proof
#1:~~}}
  {\qed\end{trivlist}}

\begin{document}

\title{Low Rank Approximation and Regression in \\ Input Sparsity Time}
\ifSTOC
\else 
\author{Kenneth L. Clarkson\\IBM Almaden \and David P. Woodruff\\IBM Almaden}
\fi
\maketitle
\begin{abstract}
We design a new distribution over $\poly(r \eps^{-1}) \times n$ matrices $S$ so that 
for any fixed $n \times d$ matrix $A$ of rank $r$, with probability at least $9/10$, 
$\norm{SAx}_2 = (1 \pm \eps)\norm{Ax}_2$ simultaneously for all $x \in \mathbb{R}^d$.
Such a matrix $S$ is called a \emph{subspace embedding}.
Furthermore, $SA$ can be computed in 
$O(\nnz(A))$
time, where $\nnz(A)$ is the number of non-zero entries of 
$A$. 
This improves over all previous subspace embeddings, which required at least $\Omega(nd \log d)$ time
to achieve this property. We call our matrices $S$ \emph{sparse embedding matrices}. 

Using our sparse embedding matrices, we obtain the fastest known algorithms for 
overconstrained least-squares
regression, low-rank approximation, approximating all leverage scores, and $\ell_p$-regression:
\begin{itemize}
\item to output an $x'$ for which 
\[
\norm{Ax'-b}_2 \leq (1+\eps)\min_x \norm{Ax-b}_2
\]
for an $n \times d$
matrix $A$ and an $n \times 1$ column vector $b$, 
we obtain an algorithm running in $O(\nnz(A)) + \tO(d^3\eps^{-2})$ time,
and another in $O(\nnz(A)\log(1/\eps)) + \tO(d^3\log(1/\eps))$ time.
(Here $\tO(f) = f \cdot \log^{O(1)}(f)$.)
\item to obtain a decomposition of an $n \times n$ matrix $A$ into a product of
an $n \times k$ matrix $L$, a $k \times k$ diagonal matrix $D$, and an $n \times k$ matrix $W$,
for which
$$\normF{A - L D  W^\top} \leq (1+\eps)\normF{A-A_k},$$ 
where $A_k$ is the best rank-$k$ approximation, our algorithm runs 
in
\[
O(\nnz(A)) + \tilde O(nk^2\eps^{-4} + k^3\eps^{-5})
\]
time.
\item to output an approximation to all leverage scores of an $n \times d$ 
input matrix $A$ simultaneously, with constant relative error,
our algorithms run in $O(\nnz(A) \log n) + \tO(r^3)$ time.
\item to output an $x'$ for which
\[
\norm{Ax'-b}_p \leq (1+\eps)\min_x \norm{Ax-b}_p
\]
for an $n \times d$
matrix $A$ and an $n \times 1$ column vector $b$, 
we obtain an algorithm running in $O(\nnz(A) \log n) + \poly(r \eps^{-1})$ time, for any 
constant $1 \leq p < \infty$. 
\end{itemize}
We optimize the polynomial factors in the above stated running times, and show various tradeoffs. 
Finally, we provide preliminary experimental results which suggest that our algorithms are of interest in practice. 
\end{abstract}

\section{Introduction}
A large body of work has been devoted to the study of fast randomized 
approximation algorithms for problems in numerical linear algebra. Several well-studied
problems in this area include
least squares regression, low rank approximation, 
and approximate computation of leverage scores. These problems have many applications
in data mining \cite{afkms01}, recommendation systems \cite{dkr02}, 
information retrieval \cite{prtv00},
web search \cite{afkm01,k99}, clustering \cite{dfkvv04,m01}, and learning mixtures of distributions \cite{ksv08,am05}. 
The use of randomization
and approximation allows one to solve these problems much faster than with deterministic
methods. 

For example, in the overconstrained least-squares regression problem, we are given an $n \times d$
matrix $A$ of rank $r$ as input, $n \gg d$, together with an $n \times 1$ column vector $b$. The goal is to 
output a vector $x'$ so that with high probability, 
$\|Ax'-b\|_2 \leq (1+\eps)\min_x \|Ax-b\|_2$.
The minimizing vector $x^*$ can be expressed in terms of the Moore-Penrose pseudoinverse $A^-$ of $A$, 
namely, $x^* = A^-b$. If $A$ has full column rank, this simplifies to $x^* = (A^\top A)^{-1}A^\top b$. 
This minimizer can be computed deterministically in $O(nd^2)$ time, but
with randomization and approximation, 
this problem can be solved in $O(nd \log d) + \poly(d \eps^{-1})$ time \cite{s06,dmms11}, 
which is much faster for $d \ll n$ and $\epsilon$ not too small. The generalization of this problem
to $\ell_p$-regression is to output a vector $x'$ so that with high probability
$\|Ax'-b\|_p \leq (1+\eps)\min_x \|Ax-b\|_p$. This can be solved exactly using convex programming,
though with randomization and approximation it is possible to achieve $O(nd \log n) + \poly(d \eps^{-1})$
time \cite{CDMMMW} for any constant $p$, $1 \leq p < \infty$. 

Another example is low rank approximation. Here we are given an $n \times n$ matrix (which can
be generalized to $n \times d$) and an input parameter $k$, and the goal is to find
an $n \times n$ matrix $A'$ of rank at most $k$ for which $\|A'-A\|_F \leq (1+\eps)\|A-A_k\|_F$,
where for an $n \times n$ matrix $B$, $\|B\|^2_F \equiv \sum_{i=1}^n \sum_{j=1}^n B_{i,j}^2$ is the
squared Frobenius norm, and $A_k\equiv \argmin_{\rank B \le k }\|A-B\|_F$.
Here $A_k$ can be computed deterministically using the singular value decomposition in $O(n^3)$ time. 
However, using randomization and approximation, this problem can be solved in  $O(\nnz(A) \cdot (k/\eps + k \log k) +
n \cdot \poly(k/\eps))$ time\cite{s06, cw09}, where  $\nnz(A)$ denotes the number of non-zero entries of $A$. 
The problem can also be solved using randomization and approximation in
$O(n^2 \log n) + n \cdot \poly(k/\eps)$ time \cite{s06}, 
which may be faster than the former for dense matrices and large $k$. 

Another problem we consider is approximating the {\it leverage scores}. 
Given an $n \times d$ matrix $A$ with $n \gg d$, one can write
$A = U \Sigma V^\top $ in its singular value decomposition, where the columns of $U$ are the left singular vectors,
$\Sigma$ is a diagonal matrix, and the columns of $V$ are the right singular vectors. Although $U$
has orthonormal columns, not much can be immediately said about the squared lengths $\|U_i\|_2^2$ of its rows.
These values are known as the leverage scores, and measure the extent to which the singular vectors of $A$ 
are correlated with the standard basis. The leverage scores are basis-independent, since
they are equal to the diagonal elements of the projection matrix onto the
span of the columns of $A$; see \cite{DMMW12} for background on leverage scores 
as well as a list of applications. 
The leverage scores will also play a crucial role in our work, as we shall see. 
The goal of approximating the leverage scores is to, 
simultaneously for each $i \in [n]$, output a constant factor approximation to $\|U_i\|_2^2$. 
Using randomization, this can be solved in $O(nd \log n + d^3 \log d \log n)$ time \cite{DMMW12}.

There are also solutions for these problems based on sampling. They either get a weaker additive error
\cite{fkv04,prtv00,am07,dkm06,dkm06a,dkm06b,dm05,rv07,drvw06},
or they get bounded relative error but are slow \cite{dv06,dmm06,dmm06b,dmm06c}.
Many of the latter algorithms
were improved independently by Deshpande and Vempala \cite{dv06} and Sarl\'os \cite{s06}, and in followup work
\cite{dmms11,ndt09,mz11}. 
%
There are also solutions based on iterative and conjugate-gradient methods, see, e.g., \cite{tb_nla},
or \cite{zf12} as recent examples. These methods repeatedly compute matrix-vector products
$Ax$ for various vectors $x$; in the most common setting, such products require $\Theta(\nnz(A))$
time. Thus the work per iteration of these methods is $\Theta(\nnz(A))$,
and the number of iterations $N$ that are performed depends on the desired accuracy,
spectral properties of $A$, numerical stability issues, and other concerns,
and can be large. A recent survey suggests that $N$ is typically $\Theta(k)$
for Krylov methods (such as Arnoldi
and Lanczos iterations) to approximate
the $k$ leading singular vectors \cite{HMT}.
One can also use some of these techniques together, for example by first obtaining a preconditioner using the
Johnson-Lindenstrauss (JL) transform, and then running an iterative method. 


While these results illustrate the power of randomization and approximation, their main drawback
is that they are not optimal. For example, for regression, ideally we could hope for $O(\nnz(A)) + \poly(d/\eps)$
time. While the $O(nd \log d) + \poly(d/\eps)$ time algorithm for least squares regression
is almost optimal for {\it dense} matrices, if $\nnz(A)\ll nd$, say $\nnz(A)=O(n)$, as commonly
occurs, this could be much worse than an $O(\nnz(A)) + \poly(d/\eps)$ time algorithm.
Similarly, for low rank approximation, the best known algorithms that are condition-independent run in 
$O(\nnz(A) (k/\eps + k \log k) + n \cdot \poly(k/\eps))$ time, while we could hope
for $O(\nnz(A)) + \poly(k/\eps)$ time. 

\subsection{Results}
We resolve the above gaps by achieving algorithms
for least squares regression, low rank approximation,
and approximate leverage scores, whose time complexities have 
a leading order term that is
$O(\nnz(A))$, sometimes up to a log factor,
with constant factors that are independent of any numerical properties of $A$.
Our results are as follows:
\begin{itemize}
\item {\bf Least Squares Regression:} We present several algorithms 
for an $n \times d$ matrix $A$ with rank $r$ and given
$\eps > 0$. One 
has running time bound of
$O(\nnz(A)\log (n/\eps) + r^3 \log^2 r + r^2\log(1/\eps))$,
stated at Theorem~\ref{thm:it reg}. (Note the logarithmic dependence on $\eps$;
a variation of this algorithm has $O(\nnz(A)\log(1/\eps)+ d^3 \log^2 d + d^2\log(1/\eps))$
running time.)
Another has 
running time $O(\nnz(A)) + \tO(d^3 \eps^{-2})$, stated at
Theorem~\ref{thm:lin reg}; note that the dependence on $\nnz(A)$ is
is linear.
We also give an algorithm for generalized (multiple-response)
regression, where $\min_X \norm{AX-B}$ is found for $B\in\R^{n\times d'}$,
in time 
\[
O(\nnz(A)\log n + r^2((r+d')\eps^{-1} + rd' + r\log^2 r + \log n));
\]
see Theorem~\ref{thm:renRegAlg}.
We also note improved results for constrained regression,
\S\ref{subsec:constrained}.
\item {\bf Low Rank Approximation:}
We achieve running time
$O(\nnz(A)) + n \cdot \poly(k(\log n)/\eps)$
to find an orthonormal $L,W\in\R^{n\times k}$ and diagonal $D\in\R^{k\times k}$ matrix
with $\norm{A-LDW^\top}_F$ within $1+\eps$ of the error of the best rank-$k$
approximation. More specifically,
Theorem~\ref{thm:SVD} gives a time bound
of
\[
O(\nnz(A)) + \tilde O(nk^2\eps^{-4} + k^3\eps^{-5}).
\]
\item {\bf Approximate Leverage Scores:} For any fixed constant $\eps > 0$, we simultaneously
$(1+\eps)$-approximate all $n$ leverage scores in 
$O(\nnz(A) \log n + r^3 \log^2 r + r^2 \log n)$ time. 
This can be generalized to sub-constant $\eps$ to achieve $O(\nnz(A) \log n) + \poly(r/\eps)$ time, 
though in the applications we are aware of, such
as coresets for regression \cite{ddhkm09}, $\eps$ is typically constant 
(in the applications of this, a general $\eps > 0$ can be achieved
by over-sampling \cite{dmm06,ddhkm09}).
\item {\bf $\ell_p$-Regression:}
For $p \in [1,\infty)$
we achieve
running time $O(\nnz(A) \log n) + \poly(r \eps^{-1})$
in 
Theorem \ref{thm:lp-running} as an immediate corollary of our results and a recent connection between 
$\ell_2$ and $\ell_p$-regression given in \cite{CDMMMW}
(for $p = 2$, the $\nnz(A) \log n$ term can be improved to $\nnz(A)$ as 
stated above).
\end{itemize}
\subsection{Techniques}
All of our results are achieved by improving the time complexity of computing what is known 
as a {\it subspace embedding}. For a given $n\times d$ matrix $A$, call $S:\R^n\mapsto\R^t$
a \emph{subspace embedding matrix} for $A$ if, for all $x\in\R^d$, $\|SAx\|_2 = (1 \pm \eps) \|Ax\|_2$.
That is, $S$ embeds the column space $C(A)\equiv \{Ax \mid x\in \R^d\}$ into $\R^t$
while approximately preserving the norms of all vectors in that subspace.

The \emph{subspace embedding problem} is to find such an embedding matrix obliviously, that is,
 to design a distribution $\pi$ over linear maps $S:\R^n\mapsto\R^t$ such that for any
 fixed $n \times d$ matrix $A$, if we choose $S \sim \pi$ then with large probability,
 $S$ is an embedding matrix for $A$. The goal is to minimize
$t$ as a function of $n, d,$ and $\eps$, while also allowing the matrix-matrix product
$S \cdot A$ to be computed quickly.

(A closely related construction, easily derived from a subspace embedding,
is an \emph{affine embedding}, involving an additional matrix $B\in\R^{n\times d'}$,
such that
\ifSTOC
\else
\[\norm{AX-B}_F\approx \norm{S(AX-B)}_F,\]
\fi
for all $X\in\R^{d\times d'}$;
see \S\ref{subsec:genAff}. These affine embeddings are used for our low-rank
approximation results, and immediately imply approximation algorithms
for constrained regression.)

By taking $S$ to be a Fast Johnson Lindenstrauss transform, one can set $t = O(d/\eps^2)$ and achieve $O(nd \log t)$ time
for $d < n^{1/2-\gamma}$ for any constant $\gamma > 0$. One can also take $S$ to be a subsampled randomized Hadamard transform, or SRHT
(see, e.g., Lemma 6 of \cite{BG}) and set $t = O(\eps^{-2} (\log d)(\sqrt{d}+\sqrt{\log n})^2)$,
to achieve $O(nd \log t)$ time.
These were the fastest known subspace embeddings achieving any value of $t$
not depending polynomially on $n$. Our main result improves this to achieve $t= \poly(d/\eps)$ for matrices $S$
for which $SA$ can be computed in $\nnz(A)$ time! Given our new subspace embedding, we plug it into known methods of solving
the above linear algebra problems given a subspace embedding as a black box.

In fact, our subspace embedding is nothing other than the {\sf CountSketch} matrix in the data stream literature \cite{ccf04}, see also
\cite{tz04}. This matrix was also studied by Dasgupta, Kumar, and Sarl\'os \cite{dks10}. Formally, $S$ has a single randomly chosen non-zero
entry $S_{h(j), j}$ in each column $j$, for a random mapping $h : [n] \mapsto [t]$.
With probability $1/2$, $S_{h(j), j} = 1$, and with probability $1/2$, $S_{h(j), j} = -1$. 

While such matrices $S$ have been studied before, the surprising fact is that they actually provide subspace embeddings.
Indeed, the usual way of proving that a random $S \sim \pi$ is a subspace embedding is to show that for any fixed vector $y \in \mathbb{R}^d$,
$\Pr[\|Sy\|_2 = (1 \pm \eps) \|y\|_2] \geq 1-\exp(-d)$. One then puts a net (see, e.g., \cite{ahk06}) on the unit vectors in the column
space  $C(A)$, and argues by a union bound that $\|Sy\|_2 = (1 \pm \eps)\|y\|_2$ for all net points $y$. This then implies, for a net that
is sufficiently fine, and using the linearity of the mapping, 
that $\|Sy\|_2 = (1 \pm \eps)\norm{y}_2$ for all vectors $y\in C(A)$. 

We stress that our choice of matrices $S$ does not preserve the norms of 
an arbitrary
set of $\exp(d)$ vectors with high probability, and so
the above approach cannot work for our choice of matrices $S$.
We instead critically use that these $\exp(d)$ vectors all come
from a $d$-dimensional subspace (namely, $C(A)$), and
therefore have a very special structure.
The structural fact we use is that there is a fixed set $H$ of size $d/\alpha$ which depends only on the
subspace, such that for any unit vector $y\in C(A)$, $H$ contains the indices
of all coordinates of $y$ larger than $\sqrt{\alpha}$ in magnitude.
The key property here is that the set $H$
is independent of $y$, or in other words, only a small set of coordinates could ever be large
as we range over all unit vectors in the subspace.
The set $H$ selects exactly the set of large leverage scores of the columns space $C(A)$!

Given this observation, by setting $t \geq K |H|^2$ for a large enough constant $K$, we have that with probability $1-1/K$, there are no
two distinct $j \neq j'$ with $j, j' \in H$ for which $h(j) = h(j')$. That is, we avoid the birthday paradox,
and the coordinates in $H$ are ``perfectly hashed'' with large
probability. Call this event $\cE$, which we condition on. 

Given a unit vector $y$ in the subspace, we can write it as $y^H + y^L$, where $y^H$ consists of $y$ with
the coordinates in $[n] \setminus H$ replaced with $0$, while $y^L$ consists of $y$ with the coordinates in $H$ replaced with $0$. We seek
to bound $$\|Sy\|_2^2 = \|Sy^H\|_2^2 + \|Sy^L\|_2^2 + 2\langle Sy^H, Sy^L \rangle.$$

Since $\cE$ occurs,
we have the isometry $\|Sy^H\|_2^2 = \|y^H\|_2^2$. Now, $\|y^L\|^2_{\infty} < \alpha $, and so we can apply Theorem 2 of 
\cite{dks10} which shows that for mappings of our form, if the input vector has small infinity norm, then $S$ preserves the norm of the
vector up to an additive $O(\eps)$ factor with high probability. Here, it suffices to set $\alpha = 1/\poly(d/\eps)$. 

Finally, we can bound
$\langle Sy^H, Sy^L \rangle$ as follows. Define $G \subseteq [n] \setminus H$ to be the set of coordinates $j$ for
which $h(j) = h(j')$ for a coordinate $j' \in H$, that is, those coordinates in $[n] \setminus H$ which ``collide'' with an element of~$H$.
Then, $\langle Sy^H, Sy^L \rangle = \langle Sy^H, Sy^{L'} \rangle$, where $y^{L'}$ is a vector which agrees with $y^L$ on coordinates
$j \in G$, and is $0$ on the remaining coordinates. By Cauchy-Schwarz, this is at most $\|Sy^H\|_2 \cdot \|Sy^{L'}\|_2$.
We have already argued that $\|Sy^H\|_2 = \|y^H\|_2 \leq 1$ for unit vectors $y$. Moreover, we can again apply Theorem 2 of \cite{dks10}
to bound $\|Sy^{L'}\|_2$, since, conditioned on the coordinates of $y^{L'}$ hashing to the set of items that the coordinates of $y^H$ hash to, they
are otherwise random, and so we again have a mapping of our form (with a smaller $t$ and applied to a smaller $n$) applied to a vector
with small infinity-norm. Therefore, $\|Sy^{L'}\|_2 \leq O(\eps) + \|y^{L'}\|_2$ with high probability. Finally, by Bernstein bounds, since the
coordinates of $y^L$ are small and $t$ is sufficiently large, $\|y^{L'}\|_2 \leq \eps$ with high probability. Hence, conditioned on event $\cE$,
$\|Sy\|_2 = (1 \pm \eps)\|y\|_2$ with probability $1-\exp(-d)$, and we can complete the argument by union-bounding over a sufficiently
fine net. 

We note that an inspiration for this work comes from work on estimating norms in a data stream with efficient update time 
by designing separate data structures for the heavy and the light components of a vector \cite{nw10,knpw11}. A key concept here is to
characterize the heaviness of coordinates in a vector space in terms of its leverage scores. 
\\\\
{\bf Optimizing the additive term:}
The above approach already illustrates the main idea behind our subspace embedding, providing the first known subspace
embedding that can be implemented in $\nnz(A)$ time. This is sufficient
to achieve our numerical linear algebra results in time $O(\nnz(A)) + \poly(d/\eps)$
for regression and $O(\nnz(A)) + n \cdot \poly(k\log(n)/\eps)$ for low rank approximation. However, 
for some applications $d, k,$ or $1/\eps$ may also be large, 
and so it is important to achieve a small degree 
in the additive $\poly(d/\eps)$ and $n \cdot \poly(k\log(n)/\eps)$ factors. 
The number of rows of the matrix $S$ is $t = \poly(d/\eps)$,
and the simplest analysis described above would give roughly 
$t = (d/\eps)^8$. We now show how to optimize this.

The first idea
for bringing this down is that the analysis of \cite{dks10} can itself 
be tightened by using that we are applying it on vectors coming 
from a subspace instead of on a set of arbitrary vectors. This involves observing that in the analysis of \cite{dks10}, if on
input vector $y$ and for every $i \in [t]$, $\sum_{j \mid h(j) = i} y_j^2$ is small then the remainder of the analysis of \cite{dks10}
does not require that $\|y\|_{\infty}$ be small. Since our vectors come from a subspace, it suffices to show that for every $i \in [t]$, 
$\sum_{j \mid h(j) = i} \|U_j\|_2^2$ is small, where $\|U_j\|_2^2$ is the $j$-th leverage score of $A$. Therefore we do not
need to perform this analysis for each $y$, but can condition on a single event, and this effectively allows us to 
increase $\alpha$ in the outline above, thereby reducing the size of $H$, and also the size of $t$ since we have
$t = \Omega(|H|^2)$. In fact, we instead follow a simpler and slightly tighter analysis of \cite{KN12} based on the Hanson-Wright
inequality. 

Another idea is that the estimation of $\|y^H\|_2$, the contribution from the 
``heavy coordinates'', is inefficient since it
requires a perfect hashing of the coordinates, which can be optimized to reduce the
additive term to $d^2 \eps^{-2} \polylog(d/\eps)$. In the worst case, there are $d$ leverage scores of value
about $1$, $2d$ of value about $1/2$, $4d$ of value about $1/4$, etc. While the top $d$ leverage
scores need to be perfectly hashed (e.g., if $A$ contains the $d \times d$ identity matrix as a
submatrix), it is not necessary that the leverage scores of smaller value, yet still larger than $1/d$,
be perfectly hashed. Allowing a small number of collisions is okay provided all vectors in the subspace have
small norm on these collisions, which just corresponds to the spectral norm of a submatrix of $A$. This 
gives an additive term of $d^2 \eps^{-2} \polylog(d/\eps)$ instead of $O(d^4 \eps^{-4})$. 
This refinement is discussed in Section \S\ref{sec:partition}.

There is yet another way to optimize the additive term to roughly $d^2 (\log n)/\eps^4$, which is useful in its own
right since the error probability of the mapping can now be made very low, namely, $1/\poly(n)$. This low error
probability bound is needed for our application to $\ell_p$-regression, see Section \ref{sec: ell_p}.
By standard balls-and-bins analyses, if we have $O(d^2/\log n)$ bins and $d^2$ balls, then with 
high probability each bin will contain $O(\log n)$ balls. We thus make $t$ roughly $O(d^2/ \log n)$ and think
of having $O(d^2/\log n)$ bins. In each bin $i$, $O(\log n)$ heavy coordinates $j$ 
will satisfy $h(j) = i$.
Then, we apply a separate JL transform 
on the coordinates that hash to each bin $i$. 
This JL
transform maps a vector $z \in \mathbb{R}^n$ to an $O((\log n) / \eps^2)$-dimensional vector $z'$ for which
$\|z'\|_2 = (1 \pm \eps) \|z\|_2$ with probability at least $1-1/\poly(n)$. Since there are only $O(\log n)$
heavy coordinates mapping to a given bin, we can put a net on all vectors on such coordinates of size only 
$\poly(n)$. We can do this for each of the 
$O(d^2 / \log n)$ bins and take a union bound. It follows 
that the $2$-norm of the vector of coordinates that hash to each bin is preserved, and so the entire
vector $y^H$ of heavy coordinates has its $2$-norm preserved. By a result of \cite{KN12}, 
the JL
transform can be implemented in $O((\log n) / \eps)$ time, giving total time $O(\nnz(A) (\log n) / \eps)$, and
this reduces $t$ to roughly $O(d^2 \log n)/\eps^4$. 

We also note that for applications such as least squares regression, it suffices to set
$\eps$ to be a constant in the subspace embedding, since we can use an approach in \cite{dmm06,ddhkm09}
which, given constant-factor 
approximations to all of the
leverage scores, can then achieve a $(1+\eps)$-approximation to least squares regression by slightly 
over-sampling rows of the adjoined matrix $A \circ b$ proportional to its leverage scores,
and solving the induced subproblem. This results in a better dependence on $\eps$. 

We can also compose our subspace embedding with a fast JL transform
to further reduce $t$ to the optimal value of about $d/\eps^2$. Since $S \cdot A$ already has small
dimensions, applying a fast JL transform is now efficient. 

Finally, we can use a recent result of \cite{ckl12} to replace most dependencies on $d$ in 
our running times for regression with a dependence on the rank $r$ of $A$, which may be smaller. 

Note that when a matrix $A$ is input that has leverage scores that are roughly equal to each other,
then the set $H$ of heavy coordinates is empty. Such a leverage score condition is assumed, for example,
in the analysis of matrix completion algorithms. For such matrices, the sketching dimension can
be made $d^2 \eps^{-2} \log(d/\eps)$, slightly improving our $d^2 \eps^{-2} \polylog(d/\eps)$ dimension above.

\subsection{Recent Related Work}
In the first version of our technical report on these ideas (July, 2012), 
the additive $\poly(k,d,1/\eps)$ terms 
were not optimized, while in the second version, the additive terms were more refined,
and results on $\ell_p$-regression for general $p$ were given,
but the analysis of sparse embeddings in \S\ref{sec:partition} was absent.
In the third version, we refined the dependence still further, with the
partitioning in \S\ref{sec:partition}.
Recently, a number of authors have told us of
followup work, all building upon our initial technical report. 

Miller and Peng showed that $\ell_2$-regression can be done
with the additive term sharpened to sub-cubic dependence
on $d$, and with linear dependence on $\nnz(A)$ \cite{MP}.
More fundamentally, they showed that a subspace embedding
can be found in $O(\nnz(A) + d^{\omega + \alpha}\eps^{-2})$ time, to
dimension
\[
O((d^{1+\alpha}\log d + \nnz(A)d^{-3})\eps^{-2});
\]
here $\omega$ is the exponent for asymptotically fast matrix multiplication,
and $\alpha>0$ is an arbitrary constant. (Some constant factors here are increasing
in $\alpha$.)

Nelson and Nguyen obtained similar results for regression, and
showed that sparse embeddings can embed into dimension $O(d^2/\eps^2)$
in $O(\nnz(A))$ time; this considerably improved on our dimension bound for that running time,
at that point (our second version),
although our current bound is within $\polylog(d/\eps)$ of their result. They also showed a
dimension bound of
$O(d^{1+\alpha})$
for $\alpha>0$, with work $O(f(\alpha)\nnz(A)\eps^{-1})$ for a particular
function of $\alpha$. Their analysis techniques are quite different from ours \cite{NN}.

Both of these papers use fast matrix multiplication to achieve
sub-cubic dependence on $d$ in applications, where our cubic term involves a JL transform,
which may have favorable properties in practice. Regarding subspace embeddings
to dimensions near-linear in $d$, note that by computing leverage scores
and then sampling based on those scores, we can obtain subspace
embeddings to $O(d\eps^{-2}\log d)$ dimensions in $O(\nnz(A)\log n) + \tO(r^3)$
time; this may be incomparable to the results just mentioned, for which the running 
times increase as $\alpha\rightarrow 0$, possibly significantly.

Paul, Boutsidis, Magdon-Ismail, and Drineas \cite{sbmd} implemented 
our subspace embeddings and found that
in the TechTC-300 matrices, a collection of 300 sparse matrices of document-term data, with an 
average of 150 to 200 rows and 15,000 columns, our subspace embeddings as used
for the projection step in their SVM classifier are about 20 times faster than the
Fast JL Transform, while maintaining the same classification accuracy. Despite this
large improvement in the time for projecting the data, further research is needed
for SVM classification, as the JL Transform empirically possesses additional properties important 
for SVM which make it faster to classify the projected data, 
even though the time to project the data using our method is faster. 

Finally, Meng and Mahoney improved on the first version of our additive terms for 
subspace embeddings, and
showed that these ideas can also be applied to $\ell_p$-regression, 
for $1 \leq p < 2$ \cite{MengMahoney}; our work on this in \S\ref{sec: ell_p}
achieves $1 \leq p < \infty$ and was done independently.
We note that our algorithms for $\ell_p$-regression require constructions
of embeddings that are successful with high probability, as we obtain
for generalized embeddings, and so some of the constructions in \cite{MP,NN} 
(as well as our non-generalized embeddings) will not yield such $\ell_p$ results. 

\subsection{Outline}

We introduce basic notation and definitions in \S\ref{sec:sparse embed}, and then the basic 
analysis in \S\ref{sec:analysis}. A more refined analysis is given in \S\ref{sec:partition}, and then
generalized embeddings, with high probability guarantees, in \S\ref{sec:generalized}.
In these sections, we generally follow the framework discussed above, splitting coordinates
of columnspace vectors into sets of ``large'' and ``small'' ones, analyzing each such set 
separately, and then bringing these analyses together.
Shifting to applications, we discuss leverage score approximation in
\S\ref{sec:leverage}, and regression in \S\ref{sec:regression}, including the use of leverage
scores and the algorithmic machinery used to estimate them, and considering
affine embeddings in \S\ref{subsec:genAff}, constrained regression in \S\ref{subsec:constrained},  and iterative methods 
in \S\ref{subsec:iterative}.
Our low-rank approximation algorithms are given in \S\ref{sec:low rank},
where we use constructions and analysis based on leverage
scores and regression.
We next apply generalized sparse embeddings to $\ell_p$-regression, in \S\ref{sec: ell_p}.
\ifSTOC
\else
Finally, in \S\ref{sec:exper}, we give some preliminary experimental results.
\fi


\section{Sparse Embedding Matrices}\label{sec:sparse embed}

We let $\norm{A}_F$ or $\norm{A}$ denote the Frobenius norm of matrix $A$,
and $\norm{A}_2$ denote the spectral norm of $A$.

Let $A \in \mathbb{R}^{n \times d}$. We assume $n > d$. 
Let $\nnz(A)$ denote the number of non-zero entries of $A$. We can assume
$\nnz(A) \geq n$ and that there are no all-zero rows or columns in $A$. 

For a parameter $t$,
we define a random linear map $\Phi D: \mathbb{R}^n \rightarrow \mathbb{R}^t$ as follows:
\begin{itemize}
\item $h : [n] \mapsto [t]$ is a random map so that for each $i\in [n]$, $h(i)=t'$ for $t'\in [t]$ with
probability $1/t$.
\item $\Phi \in \{0,1\}^{t \times n}$ is a $t \times n$ binary 
matrix with $\Phi_{h(i), i} = 1$, and all remaining entries $0$.
\item $D$ is an $n \times n$ random diagonal matrix, with each diagonal
entry independently chosen to be $+1$ or $-1$ with equal probability.
\end{itemize} 
We will refer to a matrix of the form $\Phi D$ as a {\it sparse embedding matrix}.

\section{Analysis}\label{sec:analysis}
Let $U \in \mathbb{R}^{n \times r}$ have columns that form an
orthonormal basis for the column space $C(A)$. Let $U_{1, *}, \ldots, U_{n, *}$ be the rows
of $U$, and let $u_i\equiv \norm{U_{i,*}}^2$.

It will be convenient to regard the rows of $A$ and $U$ to be re-arranged so
that the $u_i$ are in non-increasing order, so $u_1$ is largest; of course
this order is unknown  and un-used by our algorithms.

For $u\in\R^n$ and $1\le a\le b\le n$,
let $u_{a:b}$ denote the vector with $i$'th coordinate equal to $u_i$ when $i\in [a,b]$,
and zero otherwise.

Let $T > 0$ be a parameter. Throughout, we let $s\equiv \min\{i | u_i \le T\}$,
and $s'\equiv \max\{i | \sum_{s \le j\le i} u_j \le 1\}$.

We will use the notation $\Ibr{P}$, a function on event $P$, that returns 1 when
$P$ holds, and 0 otherwise.

The following variation of Bernstein's inequality\footnote{See Wikipedia entry
on Bernstein's inequalities (probability theory).}
 will be helpful.

\begin{lemma}\label{lem:Bern}
For $L,T\ge 0$ and independent random variables $X_i\in [0,T]$ with
$V \equiv \sum_i \Var[X_i]$, if $V\le LT^2/6$, then
\[
\Pr\left[\sum_i X_i \ge \sum_i \E[X_i] + LT\right] \le \exp(-L).
\]
\end{lemma}

\STOComitedproof{
\begin{proof}
Here Bernstein's inequality says that for $Y_i\equiv X_i - \E[X_i]$,
so that $\E[Y^2_i] = \Var[X_i]=V$ and $|Y_i| \le T$,
\[
\log\Pr\left[\sum_i Y_i \ge z\right] \le \frac{-z^2/2}{V + zT/3}.
\]
By the quadratic
formula, the latter is no more than $-L$ when
\[
z \ge \frac{LT}{3}(1 + \sqrt{1 + 18V/LT^2}),
\]
which holds for $z\ge LT$ and $V\le LT^2/6$.
\end{proof}
}

\subsection{Handling vectors with small entries}\label{sec:small}
We begin the analysis by considering $y\scn$ for fixed unit vectors $y\in C(A)$.
Since $\norm{y}=1$, there must be a unit vector $x$ so that $y=Ux$,
and so by Cauchy-Schwartz,
$\norm{y_i}^2\le \norm{U_{i,*}}^2\norm{x}^2 = u_i$.
This implies that $\|y\scn\|_{\infty}^2 \leq u_s$. 
We extend this to all unit vectors in subsequent sections.  

The following is similar to Lemma~6 of \cite{dks10},
and is a standard balls-and-bins analysis.

%
%

\begin{lemma}\label{lem:even hash}
For $\delta_h, T, t>0$, and $s\equiv \min\{i \mid u_i\le T\}$,
let $\cE_h$ be the event that
\[
W
	\ge \max_{j\in [t]} \sum_{\substack{i\in h^{-1}(j)\\ i\ge s}} u_i,
\]
where $W\equiv T\log(t/\delta_h) + r/t$.
If 
\[
t\ge \frac{6\norm{u\scn}^2}{T^2\log(t/\delta_h)},
\] then
$\Pr[\cE_h] \ge 1-\delta_h$.
\end{lemma}

\STOComitedproof{
\begin{proof}
We will apply Lemma~\ref{lem:Bern} to
prove that the bound holds for fixed $j\in [t]$ with failure probability
$\delta_h/t$, and then apply a union bound. 

Let $X_i$ denote the random variable $u_i \Ibr{h(i)=j, i\ge s}$.
We have $0\le X_i \le T$, $\E[X] = \sum_{i\ge s} u_i/t \le r/t$, and
$V = \sum_{i\ge s} \E[X_i^2] = \sum_{i\ge s} u_i^2/t  = \norm{u\scn}^2/t$.
Applying Lemma~\ref{lem:Bern} with $L=\log(t/\delta_h)$ gives
\[
\Pr[\sum_i X_i \ge T\log(t/\delta_h) + r/t] \le \exp(-\log(t/\delta_h)) = \delta_h/t,
\]
when $\norm{u\scn}^2/t \le LT^2/6$,
or $t\ge 6\norm{u\scn}^2/LT^2$.
\end{proof}
}

\begin{lemma}\label{lem:HW use}
For $W$ as in Lemma~\ref{lem:even hash}, suppose
the event $\cE_h$ holds. Then for unit vector $y\in C(A)$, and any $2 \leq \ell \leq 1/W$,
with failure probability $\delta_L = e^{-\ell}$,
$ | \norm{\Phi D y\scn}_2^2 - \norm{y\scn}^2 | \le K_L\sqrt{W \log(1/\delta_L)}$,
where $K_L$ is an absolute constant.
\end{lemma}

\begin{proof}
We will use the following theorem, due to Hanson and Wright.

\begin{theorem}\cite{HW}\label{thm:hw}
Let $z\in\R^n$ be a vect
or of i.i.d. $\pm 1$ random values. For any symmetric $B\in \R^{n\times n}$
and $2 \le \ell$, 
\ifSTOC
\else
\[
\E\left[ | z^\top Bz - \tr(B)| ^\ell \right] \le (CQ)^\ell,
\]
where
\[
	Q\equiv \max \{ \sqrt{\ell}\norm{B}_F, \ell \cdot\norm{B}_2\},
\]
\fi
and $C>0$ is a universal constant.
\end{theorem}

We will use Theorem~\ref{thm:hw} to prove a bound on the $\ell$'th moment
of $\norm{\Phi D y}_2^2$ for large $\ell$. Note that
$\norm{\Phi D y}^2$ can be written as $z^\top B z$, where $z$ has entries
from the diagonal of $D$, and 
$B\in\R^{n\times n}$ has $B_{ii'} \equiv y_i y_{i'} \Ibr{h(i)=h(i')}$.
Here $\tr(B)=\norm{y\scn}^2$.

Our analysis uses some ideas from the proofs for Lemmas 7 and 8 of \cite{KN12}.

Since by assumption event $\cE_h$ of Lemma \ref{lem:even hash} occurs,
and for unit $y\in C(A)$, $y_{i'}^2 \leq u_{i'}$ for all $i'$,
we have for $j\in [t]$ that
$\sum_{i' \in h^{-1}(j), i'\ge s} y^2_{i'} \le W$. Hence
\begin{align}\label{eq:B F}
\norm{B}_F^2
	  & = \sum_{i, i'\ge s} (y_i y_{i'} )^2 \Ibr{h(i')=h(i)}\nonumber
	\\ & = \sum_{i\ge s} y^2_i  \sum_{\substack{i' \in h^{-1}(h(i))\\ i'\ge s}} y^2_{i'}\nonumber
	\\ & \le \sum_{i\in [n]} y^2_i W\nonumber
	\\ & \le W.
\end{align}
For $\norm{B}_2$, observe that for given $j\in [t]$,
$z(j) \in\R^n$ with $z(j)_i = y_i \Ibr{h(i)=j, i\ge s}$ is an eigenvector of $B$
with eigenvalue $\norm{z(j)}^2$,
and the set of such eigenvectors spans the column space of $B$. It follows
that 
\[
\norm{B}_2 = \max_{j}\norm{z(j)}^2
	= \sum_{\substack{i' \in h^{-1}(j)\\ i'\ge s}} y^2_{i'} \le W.
\]
Putting this and \eqref{eq:B F} into the $Q$ of Theorem~\ref{thm:hw},
we have, 
\[
Q \le \max \{ \sqrt{\ell}\norm{B}_F, \ell \cdot\norm{B}_2\}
	\le \max\{ \sqrt{\ell}\sqrt{W}, \ell W \} =\sqrt{\ell W},
\]
where we used $\ell W \leq 1$. 
By a Markov bound applied to $| z^\top Bz - \tr(B)| ^\ell$
with $\ell = \log(1/\delta_L)$,
\begin{eqnarray*}
\Pr[ | \norm{\Phi D y\scn}_2^2 - \norm{y\scn}^2 | \ge e C \sqrt{\ell W}]
	& \le & e^{-\ell}\\
	& = & \delta_L.
\end{eqnarray*}
\end{proof}

%
%
%

\subsection{Handling vectors with large entries}\label{sec:large}

A small number of entries can be handled directly.

\begin{lemma}\label{lem:birthday}
For given $s$, let $\mathcal{E_B}$ denote the event that $h(i)\ne h(i')$ for all $i , i' < s$.
Then $\delta_B \equiv 1 -\Pr[\cE_B] \le s^2/t$.
Given event $\cE_B$, we have that for any $y$,
\[
\|y_{1:(s-1)}\|_2^2 = \|\Phi Dy_{1:(s-1)}\|_2^2.
\]
\end{lemma}

\begin{proof}
Since $\Pr[h(i)=h(i')]  = 1/t$, the probability that some such $i\ne i'$
has $h(i) = h(i')$ is at most $s^2/t$. The last claim follows by a union bound.
\end{proof}

\subsection{Handling all vectors}\label{sec:all}

We have seen that $\Phi D$ preserves the norms for vectors with small entries
(Lemma~\ref{lem:HW use}) and large entries (Lemma~\ref{lem:birthday}). Before proving
a general bound, we need to prove a bound on the ``cross terms''.

\begin{lemma}\label{lem:cross terms}
For $W$ as in Lemma~\ref{lem:even hash}, suppose
the event $\cE_h$ and $\cE_B$ hold. Then for unit vector $y\in C(A)$,
with failure probability at most $\delta_C$,
\[
|y_{1:(s-1)}^\top D\Phi^\top \Phi D y_{s:n}|\le K_C  \sqrt{W \log(1/\delta_C)},
\]
for an absolute constant $K_C$.
\end{lemma}

\ifSTOC
\else
\begin{proof}
With the event $\cE_B$, for each $i\ge s$ there is at most one $i' < s$ with $h(i)=h(i')$;
let $z_i \equiv y_{i'} D_{i'i'}$, and $z_i \equiv 0$ otherwise. 
We have for integer $p\ge 1$ using Khintchine's inequality
\begin{align*}
\E\bigg[\bigg( y_{1:(s-1)}^\top D\Phi^\top \Phi D y_{s:n}\bigg)^{2p}\bigg]^{1/p}
	  & = \E\bigg[\bigg( \sum_{i\ge s} y_i D_{ii} z_i\bigg)^{2p}\bigg]^{1/p}
	\\ & \le C_p  \sum_{i\ge s} y_i^2 z_i^2
	\\ & = C_p \sum_{i'<s} y_{i'}^2 \sum_{\substack{i \in h^{-1}(i')\\ i\ge s}} y_i^2
	\\ & \le C_p W,
\end{align*}
where $C_p\le \Gamma(p+1/2)^{1/p}=O(p)$, and the last inequality uses
the assumption that $\cE_h$ holds, and $\sum_{i' < s} y_{i'}^2 \le 1$.
Putting $p=\log(1/\delta_C)$ and applying the 
Markov inequality, we have
\[
\Pr[(y_{1:(s-1)}^\top D\Phi^\top \Phi D y_{s:n})^2 \ge e C_p W]
	\ge 1 - \exp(-p)
	= 1 - \delta_C.
\]
Therefore, with failure probability at most $\delta_C$,
we have
\[
|y_{1:(s-1)}^\top D\Phi^\top \Phi D y_{s:n}|\le K_C \sqrt{W \log(1/\delta_C)},
\]
for an absolute constant $K_C$.
\end{proof}
\fi 

\begin{lemma}\label{lem:concentrate y}
Suppose
the events $\cE_h$ and $\cE_B$ hold, and $W$ is as in Lemma~\ref{lem:even hash}.
Then for $\delta_y>0$ there is an absolute constant
$K_y$ such that, if $W\le K_y \epsilon^2/\log(1/\delta_y)$,
then for unit vector $y\in C(A)$,
with failure probability $\delta_y$,
$\norm{\Phi Dy}_2 = (1 \pm \eps)\norm{y}_2$,
when $\delta_y\le 1/2$.
\end{lemma}

\ifSTOC
\else
\begin{proof}
Assuming $\cE_h$ and $\cE_B$, we apply Lemmas~\ref{lem:birthday}, \ref{lem:HW use}, and \ref{lem:cross terms}
, and have with failure probability at most $\delta_L + \delta_C$,
\begin{align*}
| & \norm{\Phi D y}_2^2  - \norm{y}^2 |
	\\ & = | \norm{\Phi D y_{1:(s-1)}}_2^2 - \norm{y_{1:(s-1)}}^2
	\\ &\qquad	+ \norm{\Phi D y\scn}_2^2  - \norm{y\scn}^2 + 2 y_{1:(s-1)} D\Phi^\top  \Phi D y\scn |
	\\ & \le | \norm{\Phi D y\scn}_2^2 - \norm{y\scn}^2| + 0 + |2 y_{1:(s-1)} D\Phi^\top  \Phi D y\scn |
	\\ & \le  K_L\sqrt{W \log(1/\delta_L)}+ 2 K_C \sqrt{W \log(1/\delta_C)}
	\\ & \le 3\epsilon \sqrt{K_y}(K_L + K_C)
\end{align*}
for the given $W$, putting $\delta_L = \delta_C = \delta_y/2$ and assuming $\delta_y\le 1/2$.
Thus $K_y \le 1/9(K_L+K_C)^2$ suffices.
\end{proof}
\fi 

\begin{lemma}\label{lem:subspaceBound}
Suppose $\delta_{sub}>0$, $L$ is an $r$-dimensional subspace of $\mathbb{R}^n$, and $B:\mathbb{R}^n \rightarrow \mathbb{R}^k$
is a linear map. If for any fixed $x \in L$, $\|Bx\|_2^2 = (1 \pm \eps/6)\|x\|_2^2$ with probability
at least $1-\delta_{sub}$, then there is a constant $K_{sub} > 0$
for which with probability at least $1-\delta_{sub} K_{sub}^{r}$, for all $x \in L$, $\|Bx\|_2^2 = (1 \pm \eps)\|x\|_2^2$.
\end{lemma}
\ifSTOC
\else
\begin{proof}
We will need the following standard lemmas for making a net argument.
Let $S^{r-1}$ be the unit sphere in
$\R^r$ 
and let $E$ be the set of points in $S^{r-1}$ defined by
$$E =\left\{w : w \in \frac{\gamma}{\sqrt{r}} \mathbb{Z}^r, \ \|w\|_2 \leq 1 
\right\},$$
where $\mathbb{Z}^r$ is the $r$-dimensional integer lattice and $\gamma$ is a parameter.
\begin{fact}[Lemma 4 of \cite{ahk06}]\label{lem:netsize}
$|E|\le e^{cr}$ for $c = (\frac{1}{\gamma} + 2)$.
\end{fact}
\begin{fact}[Lemma 4 of \cite{ahk06}]\label{lem:netprod}
For any $r \times r$ matrix $J$, 
if for every $u, v \in E$ we have $|u^\top Jv| \leq \eps$, 
then for every unit vector $w$, we have
$|w^\top Jw| \leq \frac{\eps}{(1-\gamma)^2}$. 
\end{fact}
Let $U \in \mathbb{R}^{n \times r}$ be such that the columns
are orthonormal and the column space equals $L$. Let $I_r$ be the $r \times r$ identity matrix. Define
$J = U^TB^TBU - I_r$. Consider the set $E$ in Fact \ref{lem:netsize} and Fact \ref{lem:netprod}. 
Then, for any $x, y \in E$, we have by the statement of the lemma that with probability
at least $1-3\delta_{sub}$, $\|BUx\|_2^2 = (1 \pm \eps/6)\|Ux\|_2^2$, $\|BUy\|_2^2 = (1 \pm \eps/6)\|Uy\|_2^2$,
and $\|BU(x+y)\|_2^2 = (1 \pm \eps/6)\|U(x+y)\|_2^2 = (1 \pm \eps/6)(\|Ux\|_2^2 + \|Uy\|_2^2 + 2\langle Ux, Uy \rangle)$. 
Since $\|Ux\|_2 \leq 1$ and $\|Uy\|_2 \leq 1$, it follows that $|xJy| \leq \eps/2$. By Fact \ref{lem:netsize},
for $\gamma = 1-1/\sqrt{2}$ and sufficiently large $K_{sub}$, we have by a union bound that with probability at least 
$1-\delta_{sub} K_{sub}^r$ that 
$|xJy| \leq \eps/2$ for every $x,y \in E$. Hence, with this probability, by Fact \ref{lem:netprod}, $|w^TJw| \leq \eps$
for every unit vector $w$, which by definition of $J$ means that for all $y \in L$, $\|By\|_2^2 = (1 \pm \eps)\|y\|_2^2$.
\end{proof}
\fi 
The following is our main theorem in this section. 
\begin{theorem}\label{thm:main}
There is  $t=O((r/\epsilon)^4\log^2(r/\epsilon))$ such that with probability at least $9/10$,
$\Phi D$ is a subspace embedding matrix for $A$; that is, for all $y \in C(A)$, 
$\norm{\Phi Dy}_2 = (1 \pm \eps)\norm{y}_2$.  The embedding $\Phi D$ can be applied in $O(\nnz(A))$
time. For $s=\min\{i' \mid u_{i'}\le T\}$, where $T$ is a parameter in $\Omega(\epsilon^2/r\log(r/\epsilon))$,
it suffices if $t \ge \max\{s^2/30, r/T\}$.
\end{theorem}
\begin{proof}
For suitable $t$, $T$, and $s$,
with failure probability at most $\delta_h + \delta_B$, events $\cE_h$ and $\cE_B$ both hold.
Conditioned on this, and assuming $W$ is sufficiently small as in Lemma~\ref{lem:concentrate y},
we have with failure probability $\delta_y$ for any fixed $y\in C(A)$ that $\|\Phi D y\|_2 = (1 \pm \eps)\|y\|_2$.
Hence by Lemma \ref{lem:subspaceBound}, with failure probability
$\delta_h + \delta_B + \delta_y K_{sub}^r$, 
$\|\Phi Dy\|_2 = (1 \pm 6\eps) \|y\|_2$ for all $y \in C(A)$.
We need $\delta_h + \delta_B + \delta_y K_{sub}^r\le 1/10$, and the parameter conditions of Lemmas~\ref{lem:even hash}, 
Lemma \ref{lem:HW use}, and Lemma \ref{lem:concentrate y} holding. Listing these conditions:
\begin{enumerate}
\item $\delta_h + \delta_B + \delta_y K_{sub}^r \le 1/10$, where $\delta_B$ can be set to be $s^2/t$;
\item $u_s \le T$;
\item \label{it:s} $t \ge 6\norm{u\scn}^2/\log(t/\delta_h)T^2$;
\item $\ln(2/\delta_y) \cdot W \leq 1$ (corresponding to the condition $\ell \leq 1/W$ of Lemma \ref{lem:HW use} 
since we set $\delta_y/2 = \delta_L = e^{-\ell}$)
\item $W = T\log(t/\delta_h) + r/t \le K_y \epsilon^2/\log(1/\delta_y)$.
\end{enumerate}
We put $\delta_y =  K_{sub}^{-r}/30$, $\delta_h = 1/30$, and require $t\ge s^2/30$.
For the last condition it suffices that $T = O(\epsilon^2/r\log(t))$,
and $t = \Omega(r^2/\epsilon^2)$. The last condition implies the fourth condition for small enough constant $\eps$. 
Also, since $\norm{u\scn}^2 = \sum_{i\ge s} u_i^2 \le \sum_{i\ge s} u_i T\le rT$,
the bound for $T$ implies that $t=O((r/\epsilon)^2\log(t))$ suffices for Condition \ref{it:s}.
Thus when the leverage scores are such that
$s$ is small, $t$ can be $O((r/\epsilon)^2\log(r/\epsilon))$.
Since $\sum_i u_i = r$, $s\le r/T$ suffices, and so $t= O((r/T)^2) = O((r/\epsilon)^4\log^2(r/\epsilon))$
suffices for the conditions of the theorem.
\end{proof}

\ifSTOC
\else
\section{Partitioning Leverage Scores}\label{sec:partition}
\fi
We can further optimize our low order additive $\poly(r)$ term by refining the analysis for large leverage scores (those larger than $T$). 
We partition the scores into groups that are equal up to a constant factor, and analyze the error resulting from the 
relatively small number of collisions that may occur, using also the leverage scores to bound the error.
\ifSTOC
\else
 In what follows we 
have not optimized the $\poly(\log(r/\eps))$ factors. 

Let $q \equiv \log_2 1/T = O(\log(r/\eps))$. 
We partition the leverage scores $u_i$ with $u_i \geq T$ into
groups $G_j$, $j \in [q]$, where
$$G_j = \{i \mid 1/2^j < u_i \leq 1/2^{j-1}\}.$$
Let $\beta_j \equiv 2^{-j}$, and $n_j\equiv |G_j|$.
Since $\sum_{i=1}^n u_i = r$, we have for all $j$ that $n_j \leq r/\beta_j$.

We may also use $G_j$ to refer to the collection of rows of $U$
with leverage scores in $G_j$.

For given hash function $h$ and corresponding $\Phi$,
let $G'_j\subset G_j$ denote the collision indices of $G_j$,
those $i\in G_j$ such that $h(i)=h(i')$ for some $i'\in G_j$.
Let $k_j\equiv |G'_j|$.

First, we bound the spectral norm of a submatrix of the orthogonal basis $U$ of $C(A)$,
where the submatrix comprises rows of $G'_j$.

\subsection{The Spectral Norm}\label{subsec:spectral}
We have a matrix $B\in \R^{n_j\times r}$, with $\norm{B}_2\le 1$, and
each row of $B$ has squared Euclidean norm at least $\beta_j$
and at most $2\beta_j$, for some $j\in [q]$.

We want to bound the spectral norm of the matrix $\hat B$ whose
rows comprise those rows of $U$ in the collision set $G'_j$.
We let $t = \Theta(r^2 q^6/\epsilon^2)$
be the number of hash buckets. 
The expected number of collisions in the $t$ buckets is
$\E[|G'_j|] = \frac{{n_j \choose 2}}{t} \leq \frac{n_j^2}{2t}.$
Let $\mathcal{D}_j$ be the event that the number $k_j \equiv |G'_j| $ of such collisions in the $t$ buckets is at most $n_j^2 q^2/t$.
Let $\mathcal{D} = \cap_{j=1}^q \mathcal{D}_j$. By a 
Markov and a union bound, $\Pr[\mathcal{D}] \ge 1-1/(2q)$. We will assume that $\mathcal{D}$
occurs.

While each row in $B$ has some independent probability 
of participating in a collision, we first analyze a sampling scheme with
replacement.

We generate independent
random matrices $\hat{H}_m$ for $m\in [\ell_j]$, for a parameter $\ell_j > k_j$,
by picking $i\in [n_j]$ uniformly at random, and letting
$\hat{H}_m \equiv B_{i:}^\top B_{i:}$. Note that $\E[\hat{H}_m] = \frac{1}{n_j} B^\top B$.


Our analysis will use a special case of the version of matrix Bernstein inequalities
described by Recht.

\begin{fact}[paraphrase of Theorem 3.2 \cite{Recht}]
Let $\ell$ be an integer parameter. 
For $m\in [\ell]$, let $H_m\in\R^{r\times r}$ be independent symmetric zero-mean random matrices.
Suppose
$\rho_m^2 \equiv \norm{\E[H_m H_m]}_2$ and $M\equiv \max_{m\in [\ell]} \norm{H_m}_2$.
Then for any $\tau > 0$,
\[
\log \Pr\left[\left\|\sum_{m\in [\ell]} H_m\right\|_2 > \tau\right]
	\le \log 2r - \frac{\tau^2/2}{\sum_{m \in [\ell]}\rho_m^2 + M\tau/3}.
\]
\end{fact}

We apply this fact with $\ell = \ell_j = (4e^2)k_j + \Theta(q)$ and $H_m \equiv \hat{H}_m - \E[\hat{H}_m]$,
so that
\begin{align*}
\rho_m^2
	 &  \equiv \norm{\E[H_m H_m]}_2
	\\ & \le \left\|
		\frac{1}{n_j} \sum_{i\in [n_j]} \norm{B_{i:}}^2 B_{i:}^\top B_{i:}
		- \frac{1}{n_j^2} B^\top B B^\top B
		\right\|_2
	\\ & \le \frac{2\beta_j}{n_j} + \frac{1}{n_j^2}.
\end{align*}
Also
$M\equiv \norm{H_m}_2 \le 2\beta_j + \frac{1}{n_j}$.

Applying the above fact with these bounds for $\rho_m^2$ and $M$,
we have
\begin{align*}
\log \Pr\left[\left\|\sum_{m\in [\ell_j]} H_m\right\|_2 > \tau\right]
	  & \le \log 2r - \frac{\tau^2/2}{\sum_{m\in [\ell_j]}\rho_m^2 + M\tau/3}
	\\ & \le \log 2r -  \frac{\tau^2/2}{(2\beta_j + 1/n_j)\left(\frac{\ell_j}{n_j} + \frac{\tau}{3}\right)}.
\end{align*}
We will assume that $n_j \ge \sqrt{t}/q^2$, as discussed in lemma \ref{lem:collide norm} 
below (otherwise we have perfect hashing). With this
assumption, setting $\tau = \Theta(q(\beta_j + 1/n_j + \sqrt{r/t}))$ gives
a probability bound of $1/r$.
(Here we use that $\beta_j+1/n_j\le 2$, $\ell_j/n_j = O(q n_j/t)$, and
$\frac{n_j}{t}(\beta_j + 1/n_j) \le (r+1)/t$.)
We therefore have that with probability at least $1-1/r$,
\begin{align*}
\norm{\sum_{m\in [\ell_j]} \hat{H}_m}_2
	& = O(q(\beta_j + 1/n_j + \sqrt{r/t}) + \frac{\ell_j}{n_j}\norm{B^\top B}_2
	\\ & = O(q(\beta_j + 1/n_j + \sqrt{r/t}+ \frac{n_j}{t})),
\end{align*}
where we use that
	$\|B\|_2 \leq 1$, and use again
	$\ell_j/n_j = O(qn_j/t)$.

We can now prove the following lemma. 
\begin{lemma}\label{lem:collide norm}
With probability $1-o(1)$, for all leverage score groups $G_j$,
and for $U$ an orthonormal basis of $C(A)$, the submatrix $\hat B_j$
of $U$ consisting of rows in $G'_j$, that is, those in $G_j$ that
collide in a hash bucket with another
row in $G_j$ under $\Phi$, has squared spectral norm 
$O(q(\beta_j + 1/n_j + \sqrt{r/t}+ n_j/t))$.
\end{lemma}

\begin{proof}
Fix a $j \in [q]$. If $n_j \equiv |G_j| \leq \sqrt{t}/q^2$, then with probability $1-o(1/q)$, the items
in $G_j$ are perfectly hashed into the $t$ bins. So with probability $1-o(1)$, for all
$j \in [q]$, if $n_j \leq \sqrt{t}/q^2$, then there are no collisions. Condition on this event.

Now consider a $j \in [q]$ for which $n_j \geq \sqrt{t}/q$. Then 
$$\ell_j = (4e^2)k_j + \Theta(q) \leq n_j^2/t + O(q) \leq n_j + O(q) \leq 2n_j.$$
When sampling with replacement,
the expected number of distinct items is
$$n_j \cdot {\ell_j \choose 1} \frac{1}{n_j} \left (1- \frac{1}{n_j} \right )^{\ell_j - 1}
\geq \ell_j(1-o(1))/e^2.$$
By a standard
application of Azuma's inequality, using that $\ell_j = \Omega(q)$ is sufficiently large, 
we have that the number
of distinct items is at least $\ell_j/(4e^2)$ with probability at least $1-1/r$. By a union
bound, with probability $1-o(1)$, for all $j \in [q]$, if $n_j \geq r$, then at least 
$\ell_j/(4e^2)$ distinct items are sampled when sampling $\ell_j$ items with replacement from $G_j$. Since
$\ell_j = 4e^2k_j + O(q)$, it follows that at least $k_j$ distinct items are sampled from each $G_j$.

By the analysis above, for a fixed $j \in [q]$ we have that the submatrix of $U$ consisting of the $\ell_j$
sampled rows in $G_j$ has squared spectral norm $O(q(\beta_j + 1/n_j + \sqrt{r/t}+ n_j/t))$ with probability at least
$1-1/r$ (notice that $\|\sum_{m \in [\ell_j]} \hat{H}_m \|_2$ is the square of the spectral
norm of the submatrix of $U$ consisting of the $\ell_j$ sampled rows from $G_j$). Since the probability of 
this event is at least $1-1/r$ for a fixed $j \in [q]$, we can conclude that it holds for all $j \in [q]$
simultaneously with probability $1-o(1)$.
Finally, using that the spectral norm of a submatrix of a matrix is at most that of the matrix, we 
have that for each $j$, the squared spectral norm of a submatrix of $k_j$ random distinct rows among the $\ell_j$
sampled rows of $G_j$ from $U$ is at most $O(q(\beta_j + 1/n_j + \sqrt{r/t}+ n_j/t))$. 
\end{proof}

\subsection{Within-Group Errors}

Let $L_j\subset \R^n$ denote the set of vectors $y$ so that $y_i=0$
for $i$ not in the collision set $G'_j$,
and there is some unit $y'\in C(A)$ such that $y_i = y'_i$
for $i\in G'_j$. (Note that the error for such vectors is that
same as that for the corresponding set of vectors 
with zeros outside of $G_j$.)

In this subsection, we show that for all $y\in L_j$,
the error in estimating $\norm{y}^2$ using $y^\top D\Phi^\top \Phi D y$
is at most $O(\eps)$. 

For $y\in L_j$, the error in estimating
$\norm{y}^2$ by using $y^\top D\Phi^\top \Phi D y$ contributed
by collisions among coordinates $y_i$ for $i\in G_j$ is
\begin{equation}\label{eq:big terms}
\kappa_j \equiv \sum_{t'\in [t]} \sum_{i,i'\in h^{-1}(t')\cap G_j} y_i y_{i'} D_{ii} D_{i'i'},
\end{equation}
and we need a bound on this quantity that holds with high probability.

%
%

By a standard balls-and-bins analysis, every bucket has $O(\log t)=O(q)$
collisions, with high probability,
since $n_j \le r/T \le O(r^2/\eps^2) = O(t)$; we assume this event.

The squared Euclidean norm of the vector
of all $y_i$ that appear in the summands, that is, with $i\in G'_j$, is
at most $\beta_j + 1/n_j + \sqrt{r/t}+ n_j/t)$ by Lemma~\ref{lem:collide norm}.
Thus the squared Euclidean norm
of the vector comprising all summands in \eqref{eq:big terms}
is at most
\begin{align}
\gamma_j
	  & \equiv \sum_{t'\in [t]} \sum_{i,i'\in h^{-1}(t')\cap G_j} y_i^2 y_{i'}^2
	\\ & \le \sum_{t'\in [t]} \sum_{i\in h^{-1}(t')\cap G_j} y_i^2 O(q) 2 \beta_j\nonumber
	\\ & \le O(q^2\beta_j (\beta_j + 1/n_j + \sqrt{r/t}+ n_j/t)).\label{eq:collision norm}
\end{align}

By Khintchine's inequality, for $p\ge 1$,
\begin{align*}
\E[\kappa_j^{2p}]^{1/p}
	  & \le O(p) \gamma_j \le O(p) (q^2\beta_j (\beta_j + 1/n_j + \sqrt{r/t}+ n_j/t)),
\end{align*}
and therefore $|\kappa_j|^2$ is less than the last quantity, with failure probability
at most $4^{-p}$.

Putting $p=k'_j  \equiv \min\{r, k_j\}$, with failure probability at
most $4^{-k'_j}$, 
for any fixed vector $y\in L_j$, the squared error in estimating
$\norm{y}^2$ using the sketch of $y$ is at most
$O(k'_j(q^2\beta_j (\beta_j + 1/n_j + \sqrt{r/t}+ n_j/t))$.
Assuming the event $\cal D$ from the section above,
we have $k'_j\le \min\{r, q \cdot n_j^2q / t \}$. 
We have, using $\beta_j n_j\le r$,
\[
\frac{n_j^2 q^2}{t}(q^2\beta_j (\beta_j + 1/n_j))
	\le \frac{q^4r(r+1)}{t},
\]
and $r\cdot q^2 \beta_j n_j/t  \le q^2 r^2/t$,
and finally
\[
q^2\beta_j\sqrt{r/t}\min\{r, q^2 n_j^2/t\}
	 \le q^3\sqrt{r/t}\min\{\beta_j r, \frac{q r^2}{\beta_j t}\}
	 \le q^4 r^2/t,
\]
using $\beta_j n_j\le r$.
Putting these bounds on the terms together,
the squared error is $O(q^4r^2/t)$, or $\epsilon^2/q^2$, for
$t=\Omega(q^6r^2/\epsilon^2)$, so that the error is $O(\epsilon/q)$.

Since the dimension of $L_j$
is bounded by $k'_j$, it follows from the net argument of Lemma~\ref{lem:subspaceBound}
that for all $y\in L_j$, $\norm{Sy}^2 = \norm{y}^2 \pm O(\epsilon/q)$,
and so the total error for unit $y\in C(A)$ is $O(\epsilon)$.

We thus have the following theorem. 

\begin{theorem}\label{thm:partition within}
There is an absolute constant $C' > 0$ for which 
for any parameters $\delta_1 \in (0,1)$, $P \geq 1$, and for sparse embedding dimension 
$t = O(P (r/\eps)^2 \log^6(r/\eps))$, for all unit $y\in C(A)$,
$\sum_{j\in [q]} \norm{Sy^j} = 1 \pm C'\epsilon/P\delta_1$,
with failure probability at most $\delta_1 + O(1/\log r)$,
where $y^j$ denotes the member of $L_j$ derived from $y$.
\end{theorem}

\subsection{Handling the Cross Terms}

To complete the optimization, we must also handle the error due to ``cross terms".

Let $\delta_1 \in (0,1)$ be an arbitrary parameter. 
For $j \neq j' \in \{1, \ldots, q\}$, let the event 
$\mathcal{E}_{j,j'}$ be that the number of bins containing both an item in $G_j$
and in $G_{j'}$ is at most $\frac{n_j n_{j'} q^2}{t \delta_1}.$
Let $\mathcal{E} = \cap_{j, j'} \mathcal{E}_{j,j'}$, the
event that no pair of groups has too many inter-group collisions.
\begin{lemma}\label{lem:E}
$\Pr[\mathcal{E}] \geq 1 - \delta_1.$
\end{lemma}

\begin{proof}
Fix a $j \neq j' \in \{1, \ldots, q\}$. Then the expected number of
bins containing an item in both $G_j$ and in $G_{j'}$ is at most
$t \cdot \frac{n_j}{t} \cdot \frac{n_{j'}}{t} = \frac{n_j n_{j'}}{t},$
and so by a Markov bound the number of bins containing an item in both $G_j$ and $G_{j'}$
is at most $\frac{n_j n_{j'} q^2}{t \delta_1}$ with probability at least $1-\delta_1/q^2$. The
lemma follows by a union bound over the ${q \choose 2}$ choices of $j, j'$. 
\end{proof}

In the remainder of the analysis, we set $t = P (r/\eps)^2 q^6$ for a parameter $P \geq 1$.
\\\\
Let $\mathcal{F}$ be the event that no bin contains more than $Cq$ elements of
$\cup_{i=1}^q G_j$, where $C > 0$ is an absolute constant. 
\begin{lemma}\label{lem:F}
$\Pr[\mathcal{F}] \geq 1-1/r$. 
\end{lemma}

\begin{proof}
Observe that $|\cup_{i=1}^q G_j| = \sum_{i=1}^q n_j \leq r \sum_{i=1}^q 2^j \leq 2r^2/\eps^2.$
By standard balls and bins analysis with the given $t$, with $P \geq 1$,
with probability at least $1-1/r$ no bin contains more than $C q$ elements,
for a constant $C > 0$. 
\end{proof}

\begin{lemma}\label{lem:individual}
Condition on events $\mathcal{E}$ and $\mathcal{F}$ occurring. Consider 
any unit vector $y = Ax$ in the column space of $A$. Consider any $j \neq j' \in [q]$.
Define the vector $y^j$: $y^j_i = y_i$ for $i \in G_j$, and $y^j_i = 0$ otherwise. Then,
$$|\langle Sy^j, Sy^{j'} \rangle | = O \left (\frac{1}{P \delta_1 q^2} \right ).$$
\end{lemma}

\begin{proof}
Since $\mathcal{E}$ occurs, the number of bins containing both an item in $G_j$ and $G_{j'}$ is at most
$n_j n_{j'} q^2/(t \delta_1)$. Call this set of bins $S$. Moreover, since $\mathcal{F}$ occurs, for each bin $i \in S$,
there are at most $C \log r$ elements from $G_j$ in the bin and at most $C \log r$ elements from $G_{j'}$ in
the bin. Hence, for any $S = \Phi \cdot D$, we have, using $n_j\beta_j\le r$ for all $j$,
\begin{align*}
|\langle Sy^j, Sy^{j'} \rangle|
	  & \le \frac{n_j n_{j'} q^2}{t \delta_1} \cdot (Cq)^2 \beta_j \beta_{j'}
            \le \frac{(C q)^2 d^2 q^2}{t \delta_1} = \frac{C^2}{P \delta_1 q^2}.
\end{align*}
\end{proof}

The following is our main theorem concerning cross-terms in this section.
\begin{theorem}\label{thm:partition cross}
There is an absolute constant $C' > 0$ for which 
for any parameters $\delta_1 \in (0,1)$, $P \geq 1$, and for sparse embedding dimension 
$t = O(P (r/\eps)^2 \log^6 r)$, the event
\[\forall y = Ax \textrm{ with } \|y\|_2 = 1, 
\ \sum_{j, j' \in [q]} |\langle Sy^j, Sy^{j'} \rangle | \leq \frac{C\eps^2}{P \delta_1}
\]
occurs with failure probability at most $\delta_1 + \frac{1}{r}$,
where $y^j, y^{j'}$ are as defined in Lemma \ref{lem:individual}.
\end{theorem}
\begin{proof}
The theorem follows at once by combining Lemma \ref{lem:E}, Lemma \ref{lem:F}, and Lemma \ref{lem:individual}.
\end{proof}

\subsection{Putting it together}

Putting the bounds for within-group and cross-term errors together, and replacing
the use of Lemma~\ref{lem:birthday} in the proof of Theorem~\ref{thm:main},
we have the following theorem.

\fi 

\begin{theorem}\label{thm:partition main}
There is an absolute constant $C' > 0$ for which 
for any parameters $\delta_1 \in (0,1)$, $P \geq 1$, and for sparse embedding dimension 
$t = O(P (r/\eps)^2 \log^6 (r/\eps))$, for all unit $y\in C(A)$,
$\norm{Sy} = 1 \pm C'\epsilon/P\delta_1$,
with failure probability at most $\delta_1 + O(1/\log r)$.
\end{theorem}

\section{Generalized Sparse Embedding Matrices}\label{sec:generalized}

\ifSTOC
\else

\subsection{Johnson-Lindenstrauss transforms}
We start with a theorem of Kane and Nelson \cite{KN12}, restated here in our notation. We also
present a simple corollary that we need concerning very low dimensional subspaces. 
Let $\eps > 0, \kv = \Theta(\eps^{-1} \log (r/\eps))$, and $v = \Theta(\eps^{-1})$.
Let $B:\mathbb{R}^n \rightarrow \mathbb{R}^{v\kv}$ be defined as follows. 
We view $B$ as the concatenation (meaning, we stack the rows on top of each other) 
of matrices $\sqrt{\vk} \cdot \Phi_1 \cdot D_1, \ldots, \sqrt{\vk} \cdot \Phi_{\kv} \cdot D_{\kv}$, 
each $\Phi_i \cdot D_i$
being a linear map from $\mathbb{R}^n$ to $\mathbb{R}^v$, which is an independently
chosen sparse embedding matrix of Section \ref{sec:analysis} with
associated hash function $h_i:[n] \rightarrow [v]$. 
\begin{theorem}(\cite{KN12})\label{thm:general}
For any $\delta_{KN}, \eps > 0$,
there are $\kv = \Theta(\eps^{-1} \log (1/\delta_{KN}))$ and $v = \Theta(\eps^{-1})$
for which for any fixed $x \in \mathbb{R}^n$, a randomly chosen $B$ of the form above satisfies
$\|Bx\|_2^2 = (1 \pm \eps) \|x\|_2^2$ with probability at least $1-\delta_{KN}$. 
\end{theorem}
\begin{corollary}\label{cor:subspace}
Let $\gdelta\in (0,1)$.
Suppose $L$ is an $O(\log (r/\eps\gdelta))$-dimensional subspace of $\mathbb{R}^n$.
Let $C_{subKN} > 0$ be any constant.
Then for any $\eps\in (0,1)$, there 
are $\kv = \Theta(\eps^{-1} \log (r/\eps\gdelta))$ and $v = \Theta(\eps^{-1})$ 
such that with failure probability at most 
 $(\eps/r\gdelta)^{C_{subKN}}$, $\|By\|_2^2 = (1 \pm \eps) \|y\|_2^2$ for all $y \in L$.
\end{corollary}
\begin{proof}
We use Theorem \ref{thm:general} together with Lemma \ref{lem:subspaceBound};
for the latter, we need that for any fixed $y \in L$, $\|By\|_2^2 = (1 \pm \eps/6)\|y\|_2^2$ with probability at least
$1-\delta_{sub}$. By Theorem \ref{thm:general}, we have this for $\delta_{sub} = (\gdelta\eps/r)^{C_{KN}}$ for an arbitrarily large constant $C_{KN} > 0$.
Hence, by Lemma \ref{lem:subspaceBound}, there is a constant $K_{sub} > 0$ 
so that with probability at least
$1-(K_{sub})^{O(\log (r/\eps\gdelta))} (\gdelta\eps/r)^{C_{KN}}= 1 - (\gdelta\eps/r)^{C_{subKN}}$,
for all $y \in L$, 
$\|By\|_2^2 = (1 \pm \eps) \|y\|_2^2$. Here we use that $C_{KN} > 0$ can be made arbitrarily large, independent of $K_{sub}$. 
\end{proof} 
\subsection{The construction}
We now define a {\it generalized sparse embedding matrix} $S$. 
Let $A \in \mathbb{R}^{n \times d}$ with rank $r$. 

Let $\kv = \Theta(\eps^{-1} \log (r/\eps\gdelta))$ and
$v = \Theta(\eps^{-1})$,  be such that Theorem \ref{thm:general} 
and Corollary \ref{cor:subspace} apply
with parameters $\kv$ and $v$, for a sufficiently large constant $C_{subKN} > 0$.  
Further, let $$q \equiv C_t r \eps^{-2}(r+\log(1/\gdelta\eps)),$$ 
where $C_t > 0$ is a sufficiently large absolute constant, 
and let $t \equiv avq$. 

Let $h:[n] \rightarrow [q]$ be a random hash function. For $i = 1, 2, \ldots, q$, 
define $a_i = |h^{-1}(i)|$. Note that $\sum_{i=1}^q a_i = n$. 

We choose independent matrices 
$B^{(1)}, \ldots, B^{(q)}$, with each $B^{(i)}$ as in Theorem \ref{thm:general} 
with parameters $\kv$ and $v$. Here $B^{(i)}$ is a $v\kv \times a_i$ matrix. Finally, 
let $P$ be an $n \times n$ permutation matrix which, when applied to a matrix $A$, maps the rows
of $A$ in the set $h^{-1}(1)$
to the set of rows $\{1, 2, \ldots, a_1\}$, maps the rows of $A$ in the set $h^{-1}(2)$ to the
set of rows $\{a_1+1, \ldots, a_1+a_2\}$, and for a general $i \in [q]$,
maps the set of rows of $A$ in the set $h^{-1}(i)$ 
to the set of rows $\{a_1+a_2 + \cdots + a_{i-1}+1, \ldots, a_1 + a_2 + \cdots + a_i\}$. 

The map $S$ is defined to be the product of a block-diagonal matrix and the
matrix $P$:
\[
S \equiv \left[
\begin{matrix}
B^{(1)} & &\\
& B^{(2)} &&\\
&& \ddots & \\
&&& B^{(q)} \\
\end{matrix}
\right] \cdot P
\]

\begin{lemma}\label{lem:time}
$S \cdot A$ can be computed in $O(\nnz(A) (\log (r/\eps\gdelta)) / \eps)$ time.
\end{lemma}
\begin{proof}
As $P$ is a permutation matrix, $P \cdot A$ can be computed in $O(\nnz(A))$ time and has
the same number of non-zero entries of $A$. For each non-zero entry of $P \cdot A$,
we multiply it by $B^{(i)}$ for some $i$, which takes $O(\kv) = O(\log(r/\eps\gdelta)/\eps)$ time. Hence, the total time
to compute $S \cdot A$ is $O(\nnz(A) (\log (r/\eps\gdelta)) / \eps)$. 
\end{proof}

\subsection{Analysis}
We adapt the analysis given for sparse embedding matrices to generalized sparse embedding matrices.
Again let $U \in \mathbb{R}^{n \times r}$ have columns that form an
orthonormal basis for the column space $C(A)$. Let $U_{1, *}, \ldots, U_{n, *}$ be the rows
of $U$, and let $u_i\equiv \norm{U_{i,*}}^2$. For $\gdelta\in (0,1)$, we set the parameter:
\begin{equation}\label{eq:T JL}
T \equiv \frac{r}{C_Tq\log(t/\gdelta)} = \frac{O(\eps^2)}{\log(r/\eps\gdelta)(r + \log(1/\eps\gdelta))},
\end{equation}
where $C_T$ is a sufficiently large absolute constant. 

\subsubsection{Vectors with small entries}
Let $s\equiv \min\{i'\mid u_i \le T\}$, and for $y'\in C(A)$ of at most unit norm, let $y \equiv y'\scn$.
Since $y_i^2 \leq u_i$,
this implies that $\|y\|_{\infty}^2 \leq T$. Since $P$ is a permutation matrix, 
we have $\|Py\|_{\infty}^2 \leq T$. 

In this case, we can reduce the analysis to that of a sparse embedding matrix. Indeed, observe
that the matrix $B^{(i)}\in\R^{va\times a_i}$ is the concatenation of matrices 
$\Phi^{(i)}_1 D^{(i)}_1, \ldots, \Phi^{(i)}_{\kv} D^{(i)}_{\kv}$, where
each $\Phi^{(i)}_j D^{(i)}_j \in \R^{v\times a_i}$ is a sparse embedding matrix. Now fix a value $j \in [\kv]$ and consider
the block-diagonal matrix $N_j\in \R^{qv\times a_i}$:
\[
N_j \equiv \left[
\begin{matrix}
\Phi^{(1)}_j D^{(1)}_j & &\\
& \Phi^{(2)}_j D^{(2)}_j &&\\
&& \ddots & \\
&&& \Phi^{(q)}_j D^{(q)}_j \\
\end{matrix}
\right] \cdot P
\]
\begin{lemma}\label{lem:reduce}
$N_j$ is a random sparse embedding matrix with $qv = t/\kv$ rows and $n$ columns.
\end{lemma}
\begin{proof}
$N_j$ has a single non-zero entry in each column, and the value of this non-zero entry
is random in $\{+1, -1\}$. Hence, it remains to show that the distribution of locations
of the non-zero entries of $N_j$ is the same as that in a sparse embedding matrix. This follows
from the distribution of the values $a_1, \ldots, a_{q}$, and the definition of $P$. 
\end{proof}

\begin{lemma}\label{lem:HW JL}
Let $\gdelta\in (0,1)$.
For $j = 1, \ldots, \kv$, let $\mathcal{E}_h^j$ be the event $\mathcal{E}_h$ of Lemma 
\ref{lem:even hash}, applied to matrix $N_j$, with $\delta_h\equiv \gdelta/\kv$,
and $W\equiv T\log(qv/\delta_h) + r/qv \le 2r/C_Tq$.
Suppose $\cap_{j\in [a]} \cE_h^j$ holds.
This event has probability at least $1-\gdelta$.
Then there is an absolute constant $K_L$ such that
with failure probability at most $\delta_L$,
\[
|\norm{Sy\scn}^2 - \norm{y\scn}^2 | \le K_L\sqrt{W\log(a/\delta_L)}.
\]
\end{lemma}

\begin{proof}
We apply Lemma~\ref{lem:HW use} with $N_j$ the sparse embedding matrix $\Phi D$,
and $qv$, the number of rows of $N_j$, taking on the role of $t$ in Lemma~\ref{lem:even hash}, 
so that the parameter
$W=T\log(qv/\delta_h) + r/qv$ as in the lemma statement.
(And since $t=avq$, $qv/\delta_h  = t/\gdelta$, so $W = r/C_Tq + r/qv \le 2r/C_Tq$.)
Since $\norm{u_{s:n}}^2\le rT$, it suffices for Lemma~\ref{lem:even hash}
if $qv$ is at
least $2rT/T^2\log(t/\delta_h) = 2C_Tq$, or $v\ge 2C_T$.

With $\delta_h=\gdelta/\kv$, by a union bound $\cap_{j\in[\kv]} E_h^j$
occurs with failure probability $\gdelta$, as claimed.

We have, for given $N_j$, that with failure probability $\delta_L/a$,
$| \norm{N_j y\scn}^2 - \norm{y\scn}^2| \le K_L\sqrt{W\log(a/\delta_L)}$.
Applying a union bound, and using
$$\|Sy\scn\|_2^2 = \fvk \sum_{j=1}^{\kv} \|N_j y\scn\|_2^2,$$
the result follows.
\end{proof}

\subsubsection{Vectors with large entries}

Again, let $s\equiv \min\{i'\mid u_{i'}\le T\}$. Since $\sum_i u_i = r$,
we have
\[
s\le r/T = C_T q \log(t/\gdelta).
\]

%
%
%
The
following is a standard non-weighted balls-and-bins analysis.
\begin{lemma}\label{lem:non-weighted}
Suppose the previously defined constant $C_t > 0$ is sufficiently large. Let $\cE_{nw}$
be the event that $|h^{-1}(i) \cap [s] | \leq C_t \log (r/\eps\gdelta)$, for all $i \in [q]$.
Then $\Pr[\cE_{nw}] \ge 1-\gdelta/r$.
\end{lemma}
\begin{proof}
For any given $i \in [q]$,
\[
\E[|h^{-1}(i) \cap [s]|] = s/q \le C_T \log(t/\gdelta) = O(\log(r/\epsilon\gdelta)).
\]
Hence, by a Chernoff bound, for a constant $C_t > 0$,
$$\Pr[|h^{-1}(i) \cap [s]| > C_t \log (r/\eps\gdelta)] \leq e^{-\Theta(\log(r/\eps\gdelta))} = \frac{\gdelta}{rq},$$
The lemma now follows by a union bound over all $i \in [q]$. 
\end{proof}

\begin{lemma}\label{lem:large JL}
Assume that $\mathcal{E}_{nw}$ holds. Let $\cE_s$
be the event that for all $y\in C(A)$,
$\|Sy_{1:(s-1)}\|^2 = (1 \pm \eps/2)\|y_{1:(s-1)}\|^2$.
Then $\Pr[\cE_s] \ge 1-\gdelta/r$.
\end{lemma}

\begin{proof}
For $i = 1, 2, \ldots, q$, let $L^i$ be the at most $C_t\log (r/\eps\gdelta)$-dimensional subspace
which is the restriction of the column space $C(A)$ to coordinates $j$ with
$h(j) = i$ and $j<s$. By Corollary \ref{cor:subspace},
for any fixed $i$, with probability at least $1-(\gdelta\eps/r)^{C_{subKN}}$, for all
$y \in L^i$, $\|Sy\|^2 = (1 \pm \eps)\|y\|^2$.
By a union bound and sufficiently large $C_{subKN} > 0$, 
this holds
for all $i \in [q]$ with probability at least
$1-q(\gdelta\eps/r)^{C_{subKN}} > 1-\gdelta/r$. This condition
implies $\cE_s$, since $y_{1:(s-1)}$ can be expressed as $\sum_{i\in [q]} y^{(i)}$,
where each $y^{(i)}\in L^i$, and letting $\hat B^{(i)}$ denote the $va$ rows
of $S$ corresponding to entries from $B^{(i)}$,
\begin{align*}
\|Sy_{1:(s-1)}\|^2
	  & = \sum_{i\in [q]} \norm{\hat B^{(i)}y^{(i)}}^2
	\\ & = \sum_{i\in [q]}(1 \pm \eps)\|y^{(i)}\|^2
	\\ & = (1\pm\eps) \|y_{1:(s-1)}\|^2.
\end{align*}
A re-scaling to $\eps/2$ completes the proof.
\end{proof}

\subsection{Putting it all together}
Now consider any unit vector $y$ in $C(A)$, and write it as $y_{1:(s-1)} + y\scn$.
We seek to bound
$\langle Sy_{1:(s-1)} , Sy\scn  \rangle$. For notational convenience, define 
the block-diagonal matrix $\tilde{N}_j$ to be the matrix 
\[
\tilde{N}_j \equiv \left[
\begin{matrix}
0 & &\\
\ldots & &\\
0 & &\\
\Phi^{(1)}_j D^{(1)}_j & &\\
0 & &\\
\ldots & &\\
0 & &\\
& 0 &\\
& \ldots & \\
& 0 &\\
& \Phi^{(2)}_j D^{(2)}_j &&\\
& 0 &\\
& \ldots &\\
& 0 &\\
&& \ddots & \\
&&& 0\\
&&& \ldots \\
&&& 0\\
&&& \Phi^{(q)}_j D^{(q)}_j \\
&&& 0\\
&&& \ldots\\
&&& 0
\end{matrix}
\right] \cdot P
\]
Then $S = \sqrt{\vk} \cdot \sum_{j=1}^{\kv} \tilde{N}_j$. Notice that since the set of non-zero rows of
$\tilde{N}_j$ and $\tilde{N}_{j'}$ are disjoint for $j \neq j'$,
\begin{align}\label{eq:blockTZ}
\langle Sy_{1:(s-1)}, Sy\scn \rangle
	   & = \fvk \sum_{j=1}^{\kv} \langle \tilde{N}_j y_{1:(s-1)}, \tilde{N}_j y\scn \rangle\nonumber
	\\ & = \fvk \sum_{j=1}^{\kv} \langle N_j y_{1:(s-1)}, N_j y\scn \rangle,
\end{align}
where by Lemma \ref{lem:reduce}, each $N_j$ is a sparse embedding matrix with $qv = t/\kv$ rows and $n$ columns.

\begin{lemma}\label{lem:cross terms JL}
For $W$ as in Lemma~\ref{lem:HW JL}, and assuming events $\cap_{j=1}^{\kv} \cE_h^j$,
$\mathcal{E}_{nw}$, and $\cE_s$,
there is absolute constant $K_C$ such that with failure probability $\delta_C$,
\[
| \langle Sy_{1:(s-1)}, Sy\scn \rangle | \le K_C \sqrt{W \log(a/\delta_C)}.
\]
\end{lemma}

\begin{proof}
We generalize Lemma~\ref{lem:cross terms} slightly
to bound each summand $\langle N_j y_{1:(s-1)}, N_j y\scn \rangle$. 

For a given $j$, and for each $i\ge s$, 
let
\[
z_m \equiv \sum_{i'\in h_j^{-1}(m), i'<s} y_{i'} D^{(j)}_{i'i'},
\]
where $h_j$ is the hash function for $\Phi^{(j)}P$.
We have for integer $p\ge 1$ using Khintchine's inequality,
\begin{align*}
& \E\big[\langle N_j y_{1:(s-1)}, N_j y\scn \rangle^{2p}\big]^{1/p}
	\\ & = \E\bigg[\bigg( \sum_{i\ge s} y_i D^{(j)}_{ii} z_{h_j(i)}\bigg)^{2p}\bigg]^{1/p}
	\\ & \le C_p \sum_{i\ge s} y_i^2 z_{h_j(i)}^2
	 = C_p \sum_{m\in h_j([s-1])} z_m^2  \sum_{\substack{i\in h_j^{-1}(m)\\ i\ge s}} y_i^2
	\\ & \le C_p W V_j,
\end{align*}
where $V_j\equiv \sum_{m\in h_j^{-1}([s-1])} z_m^2$,
and $C_p\le \Gamma(p+1/2)^{1/p}=O(p)$, and the last inequality uses
the assumption that $\cE^j_h$ holds.
Putting $p=\log(a/\delta_C)$ and applying the 
Markov inequality, we have for all $j\in [a]$ that
\[
\Pr[\langle N_j y_{1:(s-1)}, N_j y\scn \rangle^2 \ge e C_p W V_j]
	\ge 1 - a\exp(-p)
	= 1 - \delta_C.
\]
Moreover, $\fvk \sum_{j\in [a]} V_j = \norm{S y_{1:(s-1)}}^2$,
which under $\cE_s$ is at most $(1+\eps/2) \norm{y_{1:(s-1)}}^2\le 1+\eps/2$.
Therefore, with failure probability at most $\delta_C$,
we have
\[
| \langle Sy_{1:(s-1)}, Sy\scn \rangle | \le K_C \sqrt{W \log(a/\delta_C)},
\]
for an absolute constant $K_C$.
\end{proof}

\fi 

The following is our main theorem in this section.
\begin{theorem}\label{thm:jlmain}
For given $\delta>0$, 
with probability at least $1-\gdelta$, for $t= O( r \eps^{-4} \log(r/\eps\gdelta)(r + \log(1/\eps\gdelta)))$,
$S$ is an embedding matrix for $A$;
that is, for all $y \in C(A)$, $\|Sy\|_2 = (1 \pm \eps)\|y\|_2$. 
$S$ can be applied to $A$ in $O(\nnz(A)\epsilon^{-1} \log (r/\gdelta))$ time.
\end{theorem}
\STOComitedproof{
\begin{proof}
Note that
\[
t=\kv v q = O([\eps^{-1}\log(r/\eps\gdelta)][\eps^{-1}][C_t r\eps^{-2}(r+\log(1/\eps\gdelta))]),
\]
yielding the bound claimed.
From Lemma~\ref{lem:HW JL}, event $\cap_{j\in[\kv]} \cE^j_h$ occurs with
failure probability at most $\gdelta$.
From Lemma~\ref{lem:non-weighted} and \ref{lem:large JL}
the joint occurrence of $\cE_{nw}$
and $\cE_s$ holds with failure probability at most $2\gdelta/r\le \gdelta$.
Given these events, from Lemmas~\ref{lem:cross terms JL} and
\ref{lem:HW JL}, we have with failure probability at most $\delta_L + \delta_C$
that
\begin{align*}
& | \norm{Sy}^2 - \norm{y}^2|
	\\  & = | \norm{Sy_{1:(s-1)}}^2 - \norm{y_{1:(s-1)}}^2
		+ \norm{Sy\scn}^2 - \norm{y\scn}^2
	\\ & \qquad\qquad	+ 2 \langle Sy_{1:(s-1)}, Sy\scn \rangle |
	\\ & \le (\eps/2) \norm{y_{1:(s-1)}}^2
		+  K_L\sqrt{W\log(a/\delta_L)}
		+ 2 K_C \sqrt{W \log(a/\delta_C)},
\end{align*}
where $W\le 2r/C_Tq$.

Setting $\delta_C = \delta_L = \gdelta K_{sub}^{-r}$,
where $K_{sub}$ is from Lemma~\ref{lem:subspaceBound},
and recalling that $\kv=O(\eps^{-1}\log(r/\eps\gdelta))$,
we have 
\[
W \log(a/\delta_L)
	\leq \frac{2r\log(a/\delta_L)}{C_T q}
	= \frac{2\eps^2 O(r + \log(1/\eps\gdelta))}{C_T (r + \log(1/\eps\gdelta))}
	\le \eps^2/C'_T,
\]
for absolute constant $C'_T$.
Using Lemma~\ref{lem:subspaceBound},
we have that with failure probability at most
$\gdelta + \gdelta +K_{sub}^r(2\gdelta K_{sub}^{-r})\le 4\gdelta$,
that
\[
| \norm{Sy}^2 - \norm{y}^2|
	\le \eps/2 +\sqrt{\eps^2/C'_T} (K_L + 2K_C)
	\le \eps
\]
for suitable choice of $C'_T$. Adjusting $\gdelta$ by a constant factor gives the result.
\end{proof}
}


\ifSTOC
\else
\section{Approximating Leverage Scores}\label{sec:leverage}
\fi
Let $A \in \mathbb{R}^{n \times d}$ with rank $r$. Let $U \in \mathbb{R}^{n \times r}$ be an
orthonormal basis for $C(A)$. In \cite{dmmw11} it was shown how to obtain 
a $(1 \pm \eps)$-approximation $u_i'$ to the leverage score $u_i$ for all $i \in [n]$, for a constant $\eps > 0$,
in time $O(nd \log n) + O(d^3 \log n \log d)$. Here we improve the running time of this task as
follows. We state the running time for constant $\eps$, though for general $\eps$ the running time
would be $O(\nnz(A)\log n) + \poly(r\eps^{-1} \log n)$. 
\begin{theorem}\label{thm:icml}
For any constant $\eps > 0$, there is an algorithm which with probability at least $2/3$, outputs
a vector $(u_1', \ldots, u_n')$ so that for all $i \in [n]$, $u_i' = (1 \pm \eps)u_i$. The running time
is %
\ifSTOC
\else
$$O(\nnz(A) \log n + r^3 \log^2 r + r^2 \log n).$$
\fi
The success probability can be amplified by independent repetition and taking the coordinate-wise
median of the vectors $u'$ across the repetitions. 
\end{theorem}
\STOComitedproof{
\begin{proof}
We first run the algorithm of Theorem 2.6 and Theorem 2.7 of \cite{ckl12}. The first theorem 
gives an algorithm which outputs the rank $r$ of $A$, while the second theorem gives an algorithm
which also outputs the indices $i_1, \ldots, i_r$ of linearly independent columns of $A$. The algorithm
takes $O(\nnz(A)\log d) + O(r^3)$ time and succeeds with probability at least $1-O(\log d)/d^{1/3}$.
Hence, in what follows, we can assume that $A$ has full rank. 

We follow the same procedure as Algorithm 1 in \cite{dmmw11}, using our improved subspace embedding.
The proof of \cite{dmmw11} proceeds by choosing a subspace embedding $\Pi_1$, computing $\Pi_1 A$,
then computing a change of basis matrix $R$ so that $\Pi_1 A R$ has orthonormal columns. The analysis
there then shows that the row norms $\|(AR)_{i, *}\|_2^2$ are equal to $u_i(1 \pm \eps)$. To obtain
these row norms quickly, an $r \times O(\log n)$ Johnson-Lindenstrauss matrix $\Pi_2$ is sampled,
and one first computes $R\Pi_2$, followed by $A(R\Pi_2)$. Using a fast Johnson-Lindenstrauss transform
$\Pi_1$, one can compute $\Pi_1 A$ in $O(nr \log n)$ time. $\Pi_1$ has $O(r \log n \log r)$ rows, and
one can compute the $r \times r$ matrix $R$ in $O(r^3 \log n \log r)$ time by computing a 
QR-factorization. Computing $R\Pi_2$
can be done in $O(r^2 \log n)$ time, and computing $A (R \Pi_2)$ can be done in $O(\nnz(A) \log n)$ time.

Our only change to this procedure is to use a different matrix $\Pi_1$, which is the composition
of our subspace embedding matrix $S$ of Theorem \ref{thm:jlmain} with parameter $t = O(r^2 \log r)$, together with a fast
Johnson Lindenstrauss transform $F$. That is, we set $\Pi_1 = F \cdot S$. Here, $F$ is an
$O(r \log^2 r) \times t$ matrix, see Section 2.3 of \cite{dmmw11} for an instantiation of $F$. 
Then, $S \cdot A$ can be computed in $O(\nnz(A) \log r)$ time by Lemma \ref{lem:time}. Moreover,
$F \cdot (SA)$ can be computed in $O(t \cdot r \log r) = O(r^3 \log^2 r)$ time. One can then 
compute the matrix $R$ above in $O(r^3 \log^2 r)$ time by computing a QR-factorization of $FSA$. Then
one can compute $R \Pi_2$ in $O(r^2 \log n)$ time, and computing $A (R \Pi_2)$ can be done in 
$O(\nnz(A) \log n)$ time.
Hence, the total time is $O(\nnz(A)\log n + r^3 \log^2 r + r^2 \log n)$ time. 

Notice that
by Theorem \ref{thm:jlmain}, with probability at least $4/5$, 
$\|Sy\|_2 = (1 \pm \eps)\|y\|_2$ for all $y \in C(A)$, and by
Lemma 3 of \cite{dmmw11}, with probability at least $9/10$, 
$\|FSy\|_2 = (1\pm \eps)\|Sy\|_2$ for all $y \in C(A)$. Hence,
$\|FSAx\|_2 = (1 \pm \eps)^2 \|Ax\|_2$ for all $x \in \mathbb{R}^d$ 
with probability at least $7/10$. There is
also a small $1/n$ probability of failure that 
$\|(AR\Pi_2)_{i, *}\|_2 \neq (1 \pm \eps)\|(AR)_{i, *}\|_2$ for
some value of $i$. Hence, the overall success probability is at least $2/3$. 

The rest of the correctness proof is identical to the analysis in \cite{dmmw11}. 
\end{proof}
} 

\section{Least Squares Regression}\label{sec:regression}
Let $A \in \mathbb{R}^{n \times d}$ and $b \in \mathbb{R}^n$ be a matrix and
vector for the regression problem: $\min_x \norm{Ax-b}_2$. We assume $n > d$. Again,
let $r$ be the rank of $A$. 
We show that with probability at least $2/3$, 
we can find an $x'$ for which
$$\norm{Ax'-b}_2 \leq (1+\eps)\min_x \norm{Ax-b}_2.$$

We will give several different algorithms.
First, we give an algorithm showing that the dependence on $\nnz(A)$ can be linear.
Next we shift to the generalized case, with multiple right-hand-sides,
and after some analytical preliminaries, give an algorithm based on sampling
using leverage scores. Finally, we discuss affine embeddings, constrained regression,
and iterative methods.

\begin{theorem}\label{thm:lin reg}
The $\ell_2$-regression problem can be solved up to a $(1+\eps)$-factor with probability at least 
$2/3$ in $O(\nnz(A) + O(d^3 \eps^{-2} \log^7(d/\eps))$ time.
\end{theorem}
\begin{proof}
By Theorem \ref{thm:main} applied to the column space $C(A\circ b)$,
where $A\circ b$ is $A$ adjoined with the vector $b$,
it suffices to compute $\Phi D A$ and $\Phi D b$ and output 
argmin$_x \norm{\Phi D Ax- \Phi D b}_2$. We use the fact that $d \geq r$, and apply
Theorem \ref{thm:partition main} with $t = O(d^2 \eps^{-2} \log^6(d/\eps))$. 

The theorem implies that with probability at least $9/10$, 
all vectors $y$ in the
space spanned by the columns of $A$ and $b$ 
have their norms preserved up to a $(1+\eps)$-factor. 
Notice that $\Phi D A$ and $\Phi D b$ can be computed in 
$O(\nnz(A))$ time. Now we have a regression problem with $d' = O(d^2 \eps^{-2} \log^6(d/\eps))$ rows
and $d$ columns. 
Using the Fast Johnson-Lindenstrauss
transform, this can be solved in $O(d' d \log (d/\eps) + d^3 \eps^{-1} \log d)$ time, see, Theorem 12 of \cite{s06}. The
success probability is at least $9/10$. This is $O(d^3 \eps^{-2} \log^7(d/\eps))$ time. 
\end{proof}

Our remaining algorithms will be stated for generalized regression.

\subsection{Generalized Regression and Affine Embeddings}

The regression problem can be slightly generalized to
\[
\min_X \normF{AX-B},
\]
where $X$ and $B$ are matrices rather than vectors. This problem,
also called \emph{multiple-response} regression,
is important in the analysis of our low-rank approximation algorithms,
and also of independent interest. Moreover, while an analysis involving
the embedding of $A\circ b$ is not significantly different than for an embedding
involving $A$ alone, this is not true for $A\circ B$: different techniques must 
be considered. This subsection gives the needed theorems
needed for analyzing algorithms for generalized regression,
and also gives a general result for \emph{affine embeddings}.

Another form of sketching matrix relies on sampling based on leverage scores; it
will be convenient to define it using sampling with replacement: for 
given sketching dimension $t$, for $m\in [t]$ let $S\in\R^{t\times n}$ have
$S_{m, z_m} \gets 1/\sqrt{tp_{z_m}}$, where $p_i \ge u_i/2r$,
and $z_m = i$ with probability $p_i$.

The following fact is due to Rudelson\cite{Rudelson}, but has since seen many proofs,
and follows readily from Noncommutative Bernstein inequalities \cite{Recht},
which are very similar to matrix Bernstein inequalities \cite{Zouzias}.

\begin{fact}\label{fact:lev embed}
For rank-$r$ $A\in\R^{n\times d}$ with row leverage scores $u_i$,
there is $t=O(r\eps^{-2}\log r)$ such that leverage-score sketching matrix $S\in \R^{t\times n}$ is
an $\epsilon$-embedding matrix for $A$.
\end{fact}

\subsection{Preliminaries}

We collect a few standard lemmas and facts in this
\ifSTOC
\else
subsection.
\fi 
\begin{lemma}\label{lem:tail}{\bf (Approximate Matrix Multiplication})
For $A$ and $B$ matrices with $n$ rows, where $A$ has $n$ columns, and given
$\epsilon>0$, there is $t=\Theta(\epsilon^{-2})$,
so that for 
a $t \times n$ generalized sparse embedding matrix $\ZZ $,
or $t\times n$ fast JL matrix, or $t\log(nd)\times n$ subsampled randomized Hadamard matrix,
or leverage-score sketching matrix for $A$ under the condition
that $A$ has orthonormal columns,
\[
\Pr[\normF{A^\top \ZZ ^\top \ZZ B - A^\top B}^2 < \epsilon^2\normF{A}^2\normF{B}^2] \ge 1-\delta,
\]
for any fixed $\delta >0$.
\end{lemma}

\ifSTOC\else
\begin{proof}
For a generalized sparse embedding
matrix with parameters $k$ and $v$, first suppose $v=1$,
so that $\ZZ $ is the embedding matrix of \S\ref{sec:sparse embed}.
Let $X = A^\top  \ZZ ^\top  \ZZ  B - AB$. Then $X_{i,j} = A_i^\top  \ZZ ^\top  \ZZ  B_j  - A_i^\top B_j$, where $A_i$ is the $i$-th column of $A$
and $B_j$ is the $j$-th column of $B$. Thorup and Zhang \cite{tz04} have shown that ${\bf E}[X_{i,j}] = 0$ 
and ${\bf Var}[X_{i,j}] = O(1/t) \norm{A_i}_2^2 \norm{B_j}_2^2.$ Consequently, 
${\bf E}[X_{i,j}^2] = {\bf Var}[X_{i,j}] = O(1/t) \cdot \norm{A_i}_2^2 \norm{B_j}_2^2,$ from which
for an appropriate $t = \Theta(\epsilon^{-2})$, the lemma follows by Chebyshev's inequality.
For $v>1$, $X_{i,j} = \frac{v}{t}\sum_{i\in [t/v]} \hat X_{i,j}$,
see \eqref{eq:blockTZ},
so that`
\[
\Var[X_{i,j}]
	= \frac{v^2}{t^2}\sum_i \Var[\hat X_{i,j}]
	\le \frac{v}{t^2}\norm{A_i}_2^2 \norm{B_j}_2^2
	\le \frac{1}{t}\norm{A_i}_2^2 \norm{B_j}_2^2,
\]
and similarly the lemma follows for the sparse embedding matrices.
The result for fast JL matrices was shown by Sarl{\'o}s\cite{s06},
and for subsampled Hadamard by Drineas et al.\cite{dmms11}, proof of Lemma~5.
(The claim also follows from norm-preserving properties of these
transforms, see \cite{kn10}.)

For leverage-score sampling, first note that
\[
A^\top \ZZ ^\top \ZZ B - A^\top B
	= \frac{1}{t} \sum_{\substack{i\in [n]\\m\in[t]}} A_{i,*}^\top B_{i,*} \left[\frac{\Ibr{z_m=i}}{p_i} - 1\right]
\]
we have $\E[A^\top \ZZ ^\top \ZZ B - A^\top B] = 0$,
and using the independence of the $z_m$,
the second moment of $\normF{A^\top \ZZ ^\top \ZZ B - A^\top B}$
is the expectation of
\begin{align*}
& \tr[ ( A^\top \ZZ ^\top \ZZ B - A^\top B)^\top (A^\top \ZZ ^\top \ZZ B - A^\top B)]
	\\ & = \frac{1}{t^2}\tr \sum_{\substack{i,i'\in [n]\\m\in[t]}}
		B_{i',*}^\top A_{i',*}A_{i,*}^\top B_{i,*}
			\left[\frac{\Ibr{z_m=i}}{p_i} - 1\right]\left[\frac{\Ibr{z_m=i'}}{p_{i'}} - 1\right],
\end{align*}
which is
\[
\frac{1}{t^2} \sum_{m\in [t]}
	\tr\left[ \left[\sum_{i\in [n]} B_{i,*}^\top A_{i,*}A_{i,*}^\top B_{i,*}\frac{1}{p_i}\right]
		- B^\top A A^\top B\right],
\]
or using the cyclic property of the trace, the fact that $p_i \ge \norm{A_{i,*}}^2/2\norm{A}^2$, and
the fact that $\tr[B^\top A A^\top B] = \norm{A^\top B}^2 \le \norm{A}^2\norm{B}^2$,
\[
\frac{1}{t} 
	\left[\sum_{i\in [n]} \norm{A_{i,*}}^2\norm{B_{i,*}}^2\frac{1}{p_i}
		- \tr[B^\top A A^\top B]\right]
	\le \frac{2}{t} \norm{A}^2\norm{B}^2,
\]
and so the lemma follows for large enough $t$ in $O(\eps^{-2})$,
by Chebyshev's inequality.
\end{proof}

\begin{fact}\label{fact:subspace jl}
Given $n\times d$ matrix $A$ of rank $k\le n^{1/2-\gamma}$ for
$\gamma>0$, and $\epsilon > 0$, an $m\times n$
fast JL matrix $\Pi$ with $m=\Theta(k/\epsilon^2)$
is a subspace embedding for $A$ with failure probability at most $\delta$,
for any fixed $\delta>0$,
and requires $O(nd\log n)$ time to apply to $A$.
\end{fact}

A similar fact holds for subsampled Hadamard transforms.

\begin{fact}\label{fact:pyth}{\bf (Pythagorean Theorem)}
If $C$ and $D$ matrices with the same number of rows and columns,
then $C^\top D=0$ implies $\normF{C+D}^2 = \normF{C}^2 + \normF{D}^2$.
\end{fact}

\begin{fact}\label{fact:normal}{\bf (Normal Equations)}
Given $n\times d$ matrix $C$, and $n\times d'$ matrix
$D$ consider the problem
\[
\min_{X\in \mathbb{R}^{d\times d'}} \normF{CX-D}^2.
\]
The solution to this problem is $X^* = C^- D$, where $C^-$ is the Moore-Penrose
inverse of $C$.  Moreover, $C^\top (CX^*-D)=0$, and so
if $c$ is any vector in the column space of $C$,
then $c^\top (CX^*-D)=0$. Using Fact~\ref{fact:pyth}, for any $X$,
\[
\normF{CX-D}^2 = \normF{C(X-X^*)}^2 + \normF{CX^*-D}^2.
\]
\end{fact}
\fi 

\subsection{Generalized Regression: Conditions}

The main theorem in this subsection is the following. It could be regarded as a
generalization of Lemma 1 of \cite{dmms11}.
\begin{theorem}\label{thm:genReg}
Suppose $A$ and $B$ are matrices with $n$ rows,
and $A$ has rank at most $r$.
Suppose
$\ZZ $ is a $t \times n$ matrix,
and the event occurs that $\ZZ $ satisfies Lemma~\ref{lem:tail} with error
parameter $\sqrt{\epsilon/r}$,
and also that
$\ZZ $ is a subspace embedding for $A$ with 
error parameter $\epsilon_0 \le 1/\sqrt{2}$.
Then
if $\tilde{Y}$ is the solution to
\begin{equation}\label{eqn:approxls}
\min_{Y} \normF{\ZZ  (A Y - B)}^2,
\end{equation}
and $Y^*$ is the solution to
\begin{equation}\label{eqn:optls}
\min_{Y} \normF{A Y - B}^2,
\end{equation}
then
$$\normF{A\tilde{Y} - B} \le (1+\epsilon) \normF{A Y^* - B}.$$
\end{theorem}
\ifSTOC\else
Before proving Theorem \ref{thm:genReg}, we will need the following lemma.  
\begin{lemma}\label{lem:betabound}
For $\ZZ , A, B, Y^*$ and $\tilde{Y}$ as in Theorem \ref{thm:genReg},
assume that $A$ has orthonormal columns. Then
$$\normF{A(\tilde{Y}-Y^*)} \leq 2\sqrt{\eps}\normF{B-AY^*}.$$
\end{lemma}

\begin{proof} The proof is in the appendix.\end{proof}
\fi 

\begin{proofof}{of Theorem \ref{thm:genReg}}
\ifSTOC
\else
Let $A$ have the thin SVD $A=U\Sigma V^\top$.
Since $U$ is a basis for $C(A)$,
there are $X^*$ and $\tilde X$ so that
$S(U\tilde X - B)= S(A\tilde Y-B)$ 
and $UX^* - B = A Y^* - B$, and therefore
$\norm{U \tilde X - B} \le (1+\eps) \norm{U X^* - B}$
implies the theorem: we can assume without loss of generality
that $A$ has orthonormal
columns. With this assumption,
and using the Pythagorean Theorem (Fact~\ref{fact:pyth})
with the normal equations (Fact~\ref{fact:normal}),
and then Lemma~\ref{lem:betabound},
\begin{align*}
\normF{A \tilde{Y} - B}^2
	   & = \normF{A Y^* - B}^2 +  \normF{A(\tilde{Y} - Y^*)}^2
	\\ & \le \normF{A Y^* - B}^2 + 4\eps\normF{AY^* - B}^2
	\\ & \le (1+4\eps)\normF{AY^* - B}^2,
\end{align*}
and taking square roots and adjusting $\eps$ by a constant factor
completes the proof.
\fi 
\end{proofof}

\subsection{Generalized Regression: Algorithm}

Our main algorithm for regression is given in the proof of the following theorem.

\begin{theorem}\label{thm:renRegAlg}
Given $A\in \R^{n\times d}$ of rank $r$, and $B\in\R^{n\times d'}$, the regression
problem $\min_Y \normF{AY-B}$ can be solved up to $\eps$ relative error
with probability at least $2/3$,
in time
\[
O(\nnz(A)\log n + r^2(r\eps^{-1} + rd' + r\log^2 r + d'\eps^{-1}  + \log n)),
\]
and obtaining a coreset of size $O(r(\eps^{-1} + \log r))$.

\end{theorem}

\begin{proof}
We estimate the leverage scores of $A$ to relative error $1/2$,
using the algorithm of
Theorem~\ref{thm:icml}, which has the side effect of finding
$r$ independent columns of $A$, so that we can assume that $d=r$.

If $U$ is a basis for $C(A)$, then
for any $X$ there is a $Y$ so that $UX=AY$,
and vice versa, so that conditions satisfied by $UX$ are satisfied
by $AY$. That is, we can (and will hereafter)
assume that $A$ has $r$ orthonormal columns, when
considering products $AY$.

We construct a leverage-score sketching matrix $S$ for $A$
with $t=O(r/\eps + r\log r)$, so that Lemma~\ref{lem:tail}
is satisfied for error parameter at most $\sqrt{\eps/r}$. With this $t$,
$S$ will also be an $\eps$-embedding matrix with
$\eps<1/\sqrt{2}$, using Lemma~\ref{fact:lev embed}.
These conditions and Theorem~\ref{thm:genReg} imply that
the solution $\tilde Y$ to $\min_Y \norm{S(AY-B)}$
has
\[
\norm{A\tilde Y - B}\le (1+\eps)\min_Y \norm{AY-B}.
\]

The running time is that for computing the leverage scores, plus
the time needed for finding $\tilde Y$, which
can be done by computing a $QR$ factorization
of $SA$ and then computing $R^{-1}Q^\top SB$,
which requires $r^3(\eps^{-1} + \log r) + r^2(\eps^{-1}+\log r) d' + r^3d'$,
and the cost bound follows.
\end{proof}

\subsection{Affine Embeddings}\label{subsec:genAff}

We also use \emph{affine embeddings}  for which a stronger condition
than Theorem~\ref{thm:genReg} is satisfied.

\begin{theorem}\label{thm:genAff}
Suppose $A$ and $B$ are matrices with $n$ rows,
and $A$ has rank at most $r$.
Suppose
$\ZZ $ is a $t \times n$ matrix,
and the event occurs that $\ZZ $ satisfies Lemma~\ref{lem:tail} with error
parameter $\eps/\sqrt{r}$, and also that
$\ZZ $ is a subspace embedding for $A$ with 
error parameter $\eps$.
Let $X^*$ be the solution of $\min_X\norm{AX-B}$, and $\tilde B \equiv AX^*-B$.
For all $X$ of appropriate shape,
\[
\norm{\ZZ(AX-B)}^2 - \norm{\ZZ\tilde B}^2  = (1\pm 2\eps) \norm{AX-B}^2 - \norm{\tilde B}^2,
\]
for $\eps\le 1/2$.
So $S$ is an affine embedding with~$2\eps$ relative error up to an additive constant.
(That is, a \emph{weak} embedding.)
If also $\norm{S\tilde B}^2 = (1\pm\eps)\norm{\tilde B}^2$, then
\begin{equation}\label{eq:aff embed}
\norm{\ZZ(AX-B)}^2 = (1\pm 3\eps) \norm{AX-B}^2,
\end{equation}
and $\ZZ$ is a $3\eps$-affine embedding.
\end{theorem}

Note that even when only the weaker first statement holds, the sketch still can be used
for optimization, since adding a constant to the objective function of an
optimization does not change the solution. Note also that 

\STOComitedproof{
\begin{proof}
If $U$ is a basis for $C(A)$, then
for any $X$ there is a $Y$ so that $UX=AY$,
and vice versa, so that conditions satisfied by $UX$ are satisfied
by $AY$. That is, we can (and will hereafter)
assume that $A$ has $r$ orthonormal columns.

Using the fact that $\norm{W}^2 = \tr W^\top W$ for any $W$,
the embedding property, the fact that $\norm{A}\le\sqrt{r}$,
and the matrix product approximation condition
of Lemma~\ref{lem:tail},
\begin{align*}
& \norm{S(AX-B)}^2 -  \norm{S\tilde B}^2
	\\ & = \norm{SA(X-X^*)+ S(AX^* - B)}^2  - \norm{S\tilde B}^2
	\\ &  = \norm{SA(X-X^*)}^2 
		- 2\tr [(X-X^*)^\top A^\top S^\top S\tilde B]
	\\ & = \norm{A(X-X^*)}^2 
	\\ & \qquad
		\pm \eps(\norm{A(X-X^*)}^2 
			+ 2 \norm{X-X^*} \norm{\tilde B}).
\end{align*}
The normal equations (Fact~\ref{fact:normal}) imply
that $\norm{AX-B}^2 = \norm{A(X-X^*)}^2 + \norm{\tilde B}^2$,
and using the observation that $(a+b)^2\le 2(a^2+b^2)$ for $a,b\in\R$,
\begin{align*}
 \norm{S(AX-B)}^2  - & \norm{S\tilde B}^2 - (\norm{AX-B}^2 - \norm{\tilde B}^2)
 	   \\ & = \pm \eps(\norm{A(X-X^*)}^2
			+ 2 \norm{X-X^*} \norm{\tilde B})
	\\ &  \le \pm \eps(\norm{A(X-X^*)} + \norm{\tilde B} )^2
	\\ & \le \pm 2\eps(\norm{A(X-X^*)}^2 + \norm{\tilde B}^2)
	\\ & = \pm 2\eps \norm{AX-B}^2,
\end{align*}
and the first statement of the theorem follows. When $\norm{S\tilde B}^2 = (1\pm\eps)\norm{\tilde B}^2$,
the second statement follows, since then
\[
\norm{S(AX-B)}^2 = (1\pm 2\eps)\norm{AX-B}^2 \pm\eps \norm{\tilde B}^2
	=  (1\pm 3\eps)\norm{AX-B}^2,
\]
using $\norm{\tilde B}\le \norm{AX-B}$ for all $X$.
\end{proof}
} 

To apply this theorem to sparse embeddings, we will need the following lemma.

\begin{lemma}\label{lem:len fixed sparse}
Let $A$ be an $n \times d$ matrix. Let $S \in \mathbb{R}^{t \times n}$ be a 
randomly chosen sparse embedding matrix for an appropriate $t = \Omega(\eps^{-2})$. 
Then with probability at least $9/10$,
$$\|SA\|_F^2 = (1 \pm \eps)\|A\|_F^2.$$
\end{lemma}
\begin{proof}
Please see the appendix.
\end{proof}

\begin{lemma}\label{lem:len fixed srht}
Let $A$ be an $n \times d$ matrix. Let $S \in \mathbb{R}^{t \times n}$ be an
SRHT matrix for an appropriate $t = \Omega(\eps^{-2} (\log n)^2)$. 
Then with probability at least $9/10$,
$$\|SA\|_F^2 = (1 \pm \eps)\|A\|_F^2.$$
\end{lemma}
\begin{proof}
Please see the appendix.
\end{proof}

\begin{theorem}\label{thm:genAff res}
Let $A$ and $B$ be matrices with $n$ rows, and $A$ has rank
at most $r$.
The following conditions hold with fixed nonzero
probability.
If $S$ is a $t\times n$ sampled randomized Hadamard transform (SRHT) matrix, there is
$t = O(\eps^{-2} [\log^2 n + (\log r)(\sqrt{r}+\sqrt{\log n})^2])$
such that $S$ is an $\eps$-affine embedding for $A$ and $B$. If $S$ is a $t\times n$
sparse embedding, there is $t=O(\eps^{-2}r^2\log^6(r/\eps))$
such that $S$ is an $\eps$-affine embedding.
If $S$ is a $t\times n$ leverage-score sampling matrix, there is $t=O(\eps^{-2} r\log r)$
such that $S$ is a weak
$\eps$-affine embedding. If the row norms of $\tilde B$ are available,
a modified leverage-score sampler is an $\eps$-embedding. (Here $\tilde B$
is as in Theorem~\ref{thm:genAff}.)
\end{theorem}

Note that none of the dimensions $t$ depend on the number of columns of $B$.

\begin{proof}
To apply Theorem~\ref{thm:genAff}, we need each given sketching matrix
to satisfy conditions on multiplicative error, subspace embedding, and
preservation of $\norm{\tilde B}$. As in that theorem, we can
assume without loss of generality that $A$ has $r$ orthonormal columns.

Regarding the multiplicative error bound of $\epsilon/\sqrt{r}$,
Lemma~\ref{lem:tail} tells us that SRHT achieves this bound for $t=O(\log(n)^2 \eps^{-2}r)$,
and the other two need $t=O(\eps^{-2}r)$.

Regarding subspace embedding, as noted in the introduction,
an SRHT matrix achieves this
for $t = O(\eps^{-2} (\log r)(\sqrt{r}+\sqrt{\log n})^2)$.
A sparse embedding requires $t=O(\eps^{-2}r^2\log^6(r/\eps))$,
as in Theorem~\ref{thm:partition main}, and leverage score samplers
need $t=O(\eps^{-2} r\log r)$, as mentioned in Fact~\ref{fact:lev embed}.

Regarding preservation of the norm of $\tilde B$, Lemma~\ref{lem:len fixed srht}
gives the claim for SRHT matrices, and %
Lemma~\ref{lem:len fixed sparse} gives the claim for
sparse embeddings, where the ``$A$'' of those lemmas
is $\tilde B$.

Thus the conditions are satisfied for Theorem~\ref{thm:genAff}
to yield the the claims for SRHT and for sparse embeddings, and for the weak condition
for leverage score samplers.

We give only a terse version of the argument for the last statement of
the theorem.
When the squared row norms $b_i \equiv \norm{\tilde B_{i,*}}^2$ of $\tilde B$ are available,
a sampler which picks row $i$ with probability $p_i = \min\{1, t b_i/\norm{\tilde B}^2\}$,
and scales that row with $1/\sqrt{tp_i}$, will yield a matrix whose
Frobenius norm will be $(1\pm 1/\sqrt{t})\norm{\tilde B}$ with high probability.
If the leverage score sampler picks rows with probability $q_i$,
create a new sampler that picks rows with probability $p'_i \equiv (p_i + q_i)/2$,
and scales by $1/\sqrt{t p'_i}$. The resulting sampler will satisfy
the norm preserving property for $\tilde B$, and also
satisfy the same properties as the leverage score sampler, up to a constant
factor. The resulting sampler is thus an $O(\eps)$-affine embedding.
\end{proof}

\subsection{Affine Embeddings and Constrained Regression}\label{subsec:constrained}

From the condition \eqref{eq:aff embed}, an affine embedding
can be used to reduce the work needed to achieve small
error in regression problems, even when there are constraints
on $X$. We consider the constraint $X\ge 0$, that the entries
of $X$ are nonnegative. The problem
$\min_{X\ge 0} \norm{AX - B}^2$,
for $B\in\R^{n\times n}$ and $A\in \R^{n\times d}$,
arises among other places as a subroutine in finding
a nonnegative approximate factorization of $B$.

For an affine embedding $S$,
\[
\min_{X\ge 0} \norm{S(AX - B)}^2
	= (1\pm \eps)\min_{X\ge 0} \norm{AX - B}^2,
\]
yielding an immediate reduction yielding a solution with relative
error $\eps$: just solve the sketched version of the problem.

From Theorem~\ref{thm:genAff res}, suitable
sketching matrices for constrained regression
include a sparse embedding, 
an SRHT matrix, or a leverage
score sampler. (The latter may not need the condition
of preserving the norm of $\tilde B$ if a high-accuracy
solver is used for the sketched solution, or if otherwise
the additive constant is not an obstacle for that solver.)

Since it's immediate that affine embeddings can be composed
to obtain an affine embedding (with a constant factor loss),
the most efficient approach might be use a sketch
that first applies a sparse embedding, and then
applies an SRHT matrix, resulting
in a sketched problem with $O(\eps^{-2}r\log(r/\eps)^2)$ rows,
and where computing the
sketch takes $O(\nnz(A) + \nnz(B)) + \tO(\eps^{-2}r^2 (d+d'))$ time,
for $B\in \R^{n\times d'}$. When $r$ is unknown, the upper
bound $r\le d$ can of course be used.

For low-rank approximation, discussed in \S\ref{sec:low rank},
we require $X$ to satisfy a rank condition; the same techniques apply.

\subsection{Iterative Methods for Regression}\label{subsec:iterative}

\ifSTOC
\else

A classical approach to finding $\min_{X} \norm{AX-B}$
is to solve the normal equations (Fact~\ref{fact:normal})
$A^\top A X = A^\top B$ via Gaussian elimination;
for $A\in\R^{n\times r}$ and $B\in\R^{n\times d'}$,
this requires $O(\min\{r\nnz(B), d'\nnz(A)\}$ to
form $A^\top B$, $O(r\nnz(A)\}$ to form $A^\top A$,
and $O(r^3 + r^2d')$ to solve the resulting linear systems.
(Another method is to factor $A=QW$,
where $Q$ has orthonormal columns and $W$ is
upper triangulation; this typically trades a slowdown for a higher-quality solution.)

Another approach to regression is to apply an iterative method
(from the general class of Krylov, CG-like methods) to a 
pre-conditioned version of the problem.  In such methods,
an estimate $x^{(m)}$ of a solution is maintained,
for iterations $m=0,1\ldots$,
using data obtained from previous iterations.
The convergence of these methods depends
on the \emph{condition number}
$\kappa(A^\top A)= \frac{\sup_{x,\norm{x}=1} \norm{Ax}^2}{\inf_{x,\norm{x}=1} \norm{Ax}^2}$
from the input matrix.
A classical result (\cite{Luen} via \cite{MSM} or Theorem 10.2.6,\cite{GvL}),
is that
\begin{equation}\label{eq:CG accuracy}
\frac{\norm{A(x^{(m)} - x^*)}^2}{\norm{A(x^{(0)} - x^*)}^2}
	\le 2\left(\frac{\sqrt{\kappa(A^\top A)} - 1}{\sqrt{\kappa(A^\top A)} + 1}\right)^m.
\end{equation}
Thus the running time of CG-like methods, such as {\tt CGNR} \cite{GvL},
depends on the (unknown)
condition number. The running time per iteration is the time needed
to compute matrix vector products $Ax$ and $A^\top v$,
plus $O(n+d)$ for vector arithmetic, or $O(\nnz(A))$.

Pre-conditioning reduces the number of iterations needed for a given accuracy:
suppose
for non-singular matrix $R$, the condition number $\kappa(R^\top A^\top AR)$
is small. Then a CG-like method applied to $AR$ would converge quickly,
and moreover for iterate $y^{(m)}$ that has error $\alpha^{(m)} \equiv \norm{ARy^{(m)} - b}$
small, the corresponding $x\gets Ry^{(m)}$ would have $\norm{Ax-b} = \alpha^{(m)}$.
The running time per iteration would have an
additional $O(d^2)$ for computing products involving $R$.

Consider the matrix $R$ obtained for leverage score approximation
in \S\ref{sec:leverage}, where a subspace embedding matrix $\Pi_1$
is applied to $A$, and $R$ is computed so that $\Pi_1 A R$  has
orthonormal columns. Since $\Pi_1$ is a subspace
embedding matrix to constant accuracy $\eps_0$,
for all unit $x\in\R^d$,
$\norm{ARx}^2=(1\pm\eps_0)\norm{\Pi_1 ARx}^2= (1\pm \eps_0)^2$.
It follows that the condition number
\[
\kappa(R^\top A^\top A R)
	\le \frac{(1+\eps_0)^2}{(1-\eps_0)^2}.
\]
That is, $AR$ is very well-conditioned. Plugging this
bound into \eqref{eq:CG accuracy}, after $m$ iterations
$\norm{AR(x^{(m)} - x^*)}^2$ is at most $2\eps_0^m$
times its starting value.

Thus starting with a solution $x^{(0)}$ with 
relative error at most 1, and applying $1+\log(1/\eps)$ iterations
of a CG-like method with $\eps_0 = 1/e$, the relative error is reduced to $\eps$
and the work is $O((\nnz(A)+ r^2)\log(1/\eps))$
(where we assume $d$ has been reduced to $r$, as in the leverage computation),
plus the work to find $R$. We have
\fi

\begin{theorem}\label{thm:it reg}
The $\ell_2$-regression problem can be solved up to a $(1+\eps)$-factor with probability at least 
$2/3$ in
\[
O(\nnz(A)\log (n/\eps) + r^3 \log^2 r + r^2\log(1/\eps))
\]
time.
\end{theorem}

\ifSTOC\else

Note that only the matrix $R$ from the leverage score computation is needed, not
the leverage scores, so the $\nnz(A)$ term in the running time 
need not have a $\log(n)$
factor; however, since reducing $A$ to $r$ columns requires that factor,
the resulting running time without that factor
is $O(\nnz(A)\log(1/\eps) + d^3\log^2 d + d^2\log(1/\eps))$,
depends on $d$.

The matrix $AR$ is so well-conditioned that a simple iterative improvement scheme
has the same running time up to a constant factor. Again start with a solution $x^{(0)}$ with 
relative error at most 1, and for $m\ge 0$,
let $x^{(m+1)} \gets x^{(m)} + R^\top A^\top (b - ARx^{(m)})$.
Then using the normal equations,
\begin{align*}
AR(x^{(m+1)} - x^*)
	  & = AR(x^{(m)} + R^\top A^\top (b - AR x^{(m)}) - x^*)
	\\ & = (AR - ARR^\top A^\top AR ) (x^{(m)} - x^*)
	\\ & = U(\Sigma - \Sigma^3)V^\top (x^{(m)} - x^*),
\end{align*}
where $AR=U \Sigma V^\top$ is the SVD of $AR$.

For all unit $x\in\R^d$,
$\norm{ARx}^2 = (1\pm \eps_0)^2$, and so
we have that all singular values $\sigma_i$ of $AR$ are $1\pm\eps_0$,
and the diagonal entries of $\Sigma - \Sigma^3$
are all at most $\sigma_i(1- (1 - \epsilon_0)^2) \le \sigma_i 3\epsilon_0$ for 
$\epsilon_0\le 1$. Hence 
\begin{align*}
\norm{AR(x^{(m+1)} - x^*)}
	 & \le 3\eps_0 \norm{AR(x^{(m)} - x^*)},
\end{align*}
and by choosing $\epsilon_0= 1/2$, say, $O(\log(1/\eps))$ iterations suffice for this
scheme also to attain $\eps$ relative error.

This scheme can be readily extended to generalized (multiple-response)
regression, using the iteration
$X^{(m+1)} \gets X^{(m)} + R^\top A^\top (B - ARX^{(m)})$.
The initialization cost then includes that of computing
$A^\top B$, which is $O(\min\{ r\nnz(B), d'\nnz(A)\})$,
where again $B\in \R^{n\times d'}$. The product $A^\top A$,
used implicitly per iteration, could be computed in $O(r\nnz(A))$,
and then applied per iteration in time $d'r^2$,
or applied each iteration in time $d'\nnz(A)$. 

That is, this method is never much worse than CG-like methods, but
comparable in running time when $d'< r$; when $d'>r$,
it is a little worse in asymptotic running time than solving the normal equations.
\fi 

\section{Low Rank Approximation}\label{sec:low rank}

This section gives algorithms for low-rank approximation, understood
using generalized regression analysis, as in earlier work such as \cite{s06,cw09}.
Let $\Delta_k \equiv \normF{A-[A]_k}$, where $[A]_k$ denotes
the best rank-$k$ approximation to $A$.
We seek low-rank matrices whose distance to $A$ is within $1+\eps$ of $\Delta_k$.

While Theorem \ref{thm:main} and Theorem \ref{thm:jlmain} 
are stated in terms of specific constant probability of success,
they can be re-stated and proven so that the failure probabilities are arbitrarily small, but still
constant. In the following we'll assume that adjustments have been done, so that the sum
of a fixed number of such failure probabilities is at most $1/5$.

We will apply embedding matrices composed of products of such matrices, so we need to check
that this operation preserves the properties we need.

\begin{fact}\label{fact:prod compose}
If $\ZZ \in \R^{t\times n}$ approximates matrix products and is a subspace
embedding with error $\epsilon$
and failure probability $\delta_\ZZ $,
and $\Pi\in \R^{\hat t\times t}$ approximates matrix products with error $\epsilon$
and failure probability $\delta_\Pi$,
then $\Pi \ZZ $ approximates matrix products with error $O(\epsilon)$
and failure probability at most $\delta_\ZZ  + \delta_\Pi$.
\end{fact}

\begin{proof}
This follows from two applications of Lemma~\ref{lem:tail}, together with
the observation that $\norm{\ZZ Ax} = (1\pm\epsilon)\norm{Ax}$ for basis vectors
$x$ implies that $\norm{\ZZ A} = (1\pm\epsilon)\norm{A}$.
\end{proof}

\begin{fact}\label{fact:embed compose}
If $\ZZ \in \R^{t\times n}$ is a subspace embedding with error $\epsilon$
and failure probability $\delta_\ZZ $,
and $\Pi\in \R^{\hat t\times t}$ is a subspace embedding with error $\epsilon$
and failure probability $\delta_\Pi$,
then $\Pi \ZZ $ is a subspace embedding with error $O(\epsilon)$
and failure probability at most $\delta_\ZZ  + \delta_\Pi$.
\end{fact}

The following lemma implies a regression algorithm that is
linear in $\nnz(A)$, but has a worse dependence in
its additive term.

\begin{lemma}\label{lem:genReg sparse}
Let $A\in \R^{n\times d}$ of rank $r$, $B\in \R^{n\times d'}$, 
and $c\equiv d+d'$.
For $\hat R\in \R^{t\times n}$ a sparse embedding matrix,
$\Pi\in \R^{t'\times t}$ a sampled randomized Hadamard matrix,
there is $t=O(r^2\log^6(r/\epsilon) + r\eps^{-1})$ and
$t'=O(r\eps^{-1}\log(r/\eps))$ such that
for $R\equiv \Pi \hat R$,
$\tilde X \equiv \argmin_X\norm{R(AX-B)}$ has
$\norm{A\tilde X - B}\le (1+\eps)\min_X\norm{AX-B}$.
The operator $R$ can be applied in $O(\nnz(A)+\nnz(B) + tc\log t)$ time.
\end{lemma}

%

\begin{theorem}\label{thm:SVD}
For $A\in\R^{n\times n}$, there is an algorithm that with failure
probability $1/10$ finds matrices $L,W\in \R^{n\times k}$
with orthonormal columns,
and diagonal $D \in\R^{k\times k}$, so that
$\norm{A-LDW^\top} \le (1+\eps)\Delta_k$.
The algorithm runs in time
\[
O(\nnz(A)) + \tilde O(nk^2\eps^{-4} + k^3\eps^{-5}).
\]
\end{theorem}

\begin{proof}
The algorithm is as follows:
\begin{enumerate}
\item Compute $AR^\top$ and an orthonormal basis $U$ for $C(AR^\top)$,
	where $R$ is as in Lemma~\ref{lem:genReg sparse} with $r=k$;
\item Compute $SU$ and $SA$ for $S$ the product of a $v'\times v$
SRHT matrix with a $v\times n$ sparse embedding,
where $v = \Theta(\eps^{-4}k^2\log^6(k/\eps))$
and $v' = \Theta(\eps^{-3}k\log^2(k/\eps))$.
(Instead of this affine embedding construction,
an alternative might use leverage score sampling,
where even the weaker claim of Theorem~\ref{thm:genAff res}
would be enough.)
\item Compute the SVD of $SU = \tilde U\tilde \Sigma \tilde V^\top$;
\item Compute the SVD $\hat U D W^\top$ of $\tilde V\tilde \Sigma^- [\tilde U^\top SA]_k$,
	where again $[Z]_k$ denotes the best rank-$k$ approximation to matrix $Z$;
\item Return $L=U\hat U$, $D$, and $W$.
\end{enumerate}

{\bf Running time.}
Computing $AR^\top$ in the first step takes
$O(\nnz(A) + \tO(nk(k+\eps^{-1}))$ time,
and then $\tilde O(n(k/\eps)^2)$ to compute the $n\times O(\eps^{-1}k\log(k/\eps))$
matrix $U$. 
Computing $SU$ and $SA$ requires
$O(\nnz(A)) + \tO(n k^2\eps^{-4})$ time.
Computing the SVD of
the $\tilde O(k\eps^{-3}) \times \tilde O(k\eps^{-1})$  matrix $SU$
requires $\tilde O(k^3\eps^{-5})$. Computing $\tilde U^\top SA$
requires $\tilde O(nk^2\eps^{-4})$ time.
Computing the SVD of the
$\tilde O(k\eps^{-1})\times n$ matrix of the next step requires
$\tilde O(nk^2\eps^{-2})$ time, as does computing $U\hat U$.

\paragraph{Correctness}
Apply Lemma~\ref{lem:genReg sparse} with $A$ of that lemma
mapping to $A_k^\top$ and $B$ mapping to $A^\top$.
Taking transposes, this implies that with small fixed failure probability,
$\tilde Y \equiv AR^\top (A_kR^\top)^-$
has
\[
\norm{\tilde Y A_k - A} \le (1+\epsilon) \min_Y\norm{ YA_k - A} = (1+\epsilon)\Delta_k,
\]
and so
\begin{align}\label{eq:AR good}
\min_{X, \rank X=k} \norm{AR^\top X - A}
	  & \le \norm{AR^\top (A_kR^\top)^-A_k - A} \nonumber
	  \\ & \le (1+\epsilon)\Delta_k.
\end{align}
Since $U$ is a basis for $C(AR^\top)$,
\[
(1+\eps)\min_{X, \rank X=k} \norm{UX - A}
	\le (1+\eps) \min_{X, \rank X=k} \norm{AR^\top X - A}.
\]
With the given construction of $S$,
Theorem~\ref{thm:genAff res} applies (twice), with $AR^\top$ taking the role
of $A$, and $A$ taking the role of $B$, so that $S$ is
an $\eps$-affine embedding, after adjusting constants.
It follows that
for $\tilde X\equiv\argmin_{X, \rank X=k} \norm{S(UX - A)}$,
\begin{align*}
\norm{U\tilde X - A}
	  & \le (1+\eps)\min_{X, \rank X=k} \norm{UX - A}
	\\ & \le (1+\eps) \min_{X, \rank X=k} \norm{AR^\top X - A}
	\\ & \le (1+\eps)^2\Delta_k,
\end{align*}
using \eqref{eq:AR good}.
From lemma 4.3 of \cite{cw09},
the solution to
\[
\min_{X,\rank X=k}\norm{\tilde U X - SA}
\]
is $\hat X = [\tilde U^\top SA]_k$, where this denotes the best rank-$k$
approximation to $\tilde U^\top SA$.
It follows that $\tilde X = \tilde V\tilde \Sigma^- \hat X$ is a solution to
to $\min_{X, \rank X=k} \norm{S(UX - A)}$.
Moreover, the rank-$k$ matrix $U\tilde X = LDW^\top$
has $\norm{LDW^\top- A}\le (1+\eps)^2\Delta_k$,
and $L$, $D$, and $W$ have the properties promised.
\end{proof}

\section{$\ell_p$-Regression for any $1 \leq p < \infty$} \label{sec: ell_p}
Let $A \in \mathbb{R}^{n \times d}$ and $b \in \mathbb{R}^n$ be a matrix and
vector for the regression problem: $\min_x \norm{Ax-b}_p$. We assume $n > d$. 
Let $r$ be the rank of $A$. 
We show that with probability at least $2/3$, 
we can quickly find an $x'$ for which
$$\norm{Ax'-b}_p \leq (1+\eps)\min_x \norm{Ax-b}_p.$$
Here $p$ is any constant in $[1, \infty)$. 

This theorem is an immediate corollary of Theorem \ref{thm:jlmain} 
and the construction given in 
section 3.2 of \cite{CDMMMW}, which shows how to solve $\ell_p$-regression given a subspace embedding
(for $\ell_2$) as a black box.
\ifSTOC
\else
We review the construction of \cite{CDMMMW} below for completeness. 

As in the proof of Theorem \ref{thm:icml}, in $O(\nnz(A) \log d) + O(r^3)$ time we can 
replace the input matrix $A$ with a new matrix with the same column space of $A$ and 
full column rank,
where $r$ is rank of $A$. We therefore assume $A$ has full rank in what follows. 

Let $w = \Theta(r^6 \log n (r + \log n))$ and assume $w \mid n$. 
Split $A$ into $n/w$ matrices $A_1, \ldots, A_{n/w}$,
each $w \times r$, so that $A_i$ is the submatrix of $A$ indexed by the $i$-th block of 
$w$ rows. 

We invoke Theorem \ref{thm:jlmain} with the parameters $n = w$, $r$, $\eps = 1/2$, 
and $\delta = 1/(100n)$, choosing a generalized sparse embedding matrix matrix $S$ with 
$t = O(r \log n(r + \log n))$ rows. Theorem \ref{thm:jlmain} has the guarantee that for each fixed $i$,
$SA_i$ is a subspace embedding with probability at least $1-\delta$. It follows by a union
bound that with probability at least $1- 1/(100w)$, for all $i \in [n/w]$, $SA_i$ 
is a subspace embedding. We condition on this event occurring. 

%
%
%
%
%
Consider the matrix
$F \in\R^{nt/w \times n}$, which 
is a block-diagonal matrix comprising $n/w$ blocks along the diagonal.
Each block is the $t \times w$ matrix $S$ given above. 
\[
F \equiv \left[
\begin{matrix}
S & &\\
& S &&\\
&& \ddots & \\
&&& S \\
\end{matrix}
\right]
\]
We will need the following theorem.
\begin{theorem}[Theorem 5 of \cite{ddhkm09}, restated]
\label{thm:basisOld}
Let $A$ be an $n \times r$ matrix, and let $p \in [1, \infty)$.
Then there exists an $(\alpha, \beta, p)$-well-conditioned basis for
the column space of $A$ such that if $p < 2$, then $\alpha = r^{1/2 + 1/p}$ and $\beta = 1$; 
if $p = 2$, then $\alpha = r^{1/2}$ and $\beta = 1$, and if $p > 2$ then 
$\alpha = r^{1/2 + 1/p}$ and $\beta = r^{1/2 - 1/p}$. An $r \times r$ change of basis matrix
$U$ for which $A \cdot U$ is a well-conditioned basis can be
computed in $O(nr^5 \log n)$ time. 
\end{theorem}
\noindent
The specific conditions satisfied by a well-conditioned basis are given (and used)
in the proof of the theorem below. We use the following algorithm {\sf Condition}($A$) 
given a matrix $A \in \mathbb{R}^{n \times r}$:
\begin{enumerate}
\item Compute $FA$;
\item Apply Theorem \ref{thm:basisOld} to $FA$ to obtain an $r \times r$ change of basis
matrix $U$ so that $FAU$ is an $(\alpha, \beta, p)$-well-conditioned basis of the
column space of matrix $FA$;
\item Output $AU/(r\gamma_p)$,
where $\gamma_p \equiv \sqrt{2} t^{1/p-1/2}$ for $p\le 2$,
and $\gamma_p \equiv \sqrt{2} w^{1/2-1/p}$ for $p\ge 2$.
\end{enumerate}
The following lemma is the analogue of that in \cite{CDMMMW} proved for the Fast Johnson Lindenstauss
Transform. However, the proof in \cite{CDMMMW} only used that the Fast Johnson Lindenstrauss Transform is a subspace
embedding. We state it here with our new parameters,
\ifSTOC
omitting the proof in this version.
\else %
and give the analogous proof in the Appendix for completeness.
\fi
\begin{lemma}\label{lem:lpl2}
With probability at least $1-1/(100w)$, the output $A U/(r\gamma_p)$ of {\sf Condition}$(A)$
is guaranteed to be a basis that is $(\alpha, \beta \sqrt{3} r (tw)^{|1/p-1/2|} , p)$-well-conditioned,
that is, an $(\alpha, \beta \cdot \poly(\max(r, \log n)) , p)$-well-conditioned basis.
The time to compute $U$ is $O(\nnz(A) \log n) + \poly(r \eps^{-1})$. 
\end{lemma}
The following text is from \cite{CDMMMW}, we state it here for completeness. 
A well-conditioned basis can be used 
to solve $\ell_p$ regression problems, via an algorithm based on sampling the rows
of $A$ with probabilities proportional to the norms of the rows of the corresponding
well-conditioned basis. This entails
using for speed a second projection $\Pi_2$ applied to $AU$ on the right
to estimate the row norms, where $\Pi_2$ can be an $O(r) \times O(\log n)$ matrix of 
i.i.d. normal random variables, which is the same as is done in \cite{dmmw11}. 
This allows fast estimation of the $\ell_2$ norms of the rows of $AU$;
however, we need the $\ell_p$ norms of those rows, which we thus know up
to a factor of $r^{|1/2 - 1/p|}$.
We use these norm estimates
in the sampling algorithm of \cite{ddhkm09}; as discussed for
the running time bound of that paper, Theorem 7, this algorithm samples
a number of rows proportional to $r(\alpha\beta)^p$,
when an $(\alpha, \beta, p)$-well-conditioned
basis is available. This factor, together with
a sample complexity increase of $r^{p|1/2-1/p|} = r^{|p/2-1|}$ needed to compensate
for error due to using $\Pi_2$, gives a sample complexity increase for our algorithm
over that of \cite{ddhkm09} of a factor
of
\[
[r^{|p/2 - 1|}]r^{p+1}(tw)^{|p/2 - 1|} = \max(r, \log n)^{O(p)},
\]
while the leading term in the complexity (for $n\gg r$)
is reduced from $O(nr^5\log n)$ to $O(\nnz(A)\log n)$. 

Observe that if $r < \log n$,
then $\poly(r \eps^{-1} \log n)$ is less than $n \log n$, which is $O(\nnz(A) \log n)$. 
Hence, the overall time complexity is $O(\nnz(A) \log n) + \poly(r \eps^{-1})$. 

We adjust Theorem 4.1 of \cite{ddhkm09} 
and obtain the following.
\fi 

\begin{theorem}\label{thm:lp-running}
Given $\epsilon\in(0,1)$, a constant $p \in [1, \infty)$, 
$A\in\R^{n\times d}$ and $b\in\R^{n}$,
there is a sampling algorithm for $\ell_p$ regression that
constructs a coreset specified by a diagonal sampling matrix $D$, and 
a solution vector $\hat{x} \in \R^d$ that minimizes the weighted regression objective
$\|D(Ax-b)\|_p$.  The solution $\hat{x}$ satisfies, with probability
at least $1/2$, the relative error bound that
$\|A \hat{x}-b\|_p\le (1+\epsilon)\|Ax-b\|_p$
for all $x\in\R^d$.
Further, with probability
$1-o(1)$, the entire algorithm to construct $\hat{x}$  runs in time
$$O\left(\nnz(A)\log n \right ) + \poly(r \eps^{-1}).$$
\end{theorem}

\ifSTOC\else
\section{Preliminary Experiments}\label{sec:exper}

Some preliminary experiments show that a low-rank approximation technique that is a simplified
version of these algorithms is promising, and in practice may perform much better
than the general bounds of our results.

Here we apply the algorithm of Theorem~\ref{thm:SVD}, except that we skip the
randomized Hadamard and simply use a sparse embedding $\hat R$ 
and leverage score sampling. We compare the Frobenius error of the resulting $LDW^\top$
with that of the best rank-$k$ approximation.

%
%
%
In our experiments, the matrices tested
are $n\times d$.

The resulting low-rank approximation
was tested for $t_R$ (the number of columns of $\hat R$)
taking values of the form$\lfloor 1.6^z - 0.5\rfloor$, for integer
$z\ge 1$, while $t_R\le d/5$.
The number $t_S$ of rows of $S$ was chosen such that the condition number of $SU$ was at most $1.2$. (Since $U$ has orthogonal columns, its condition number is 1, so a large enough leverage
score sample will have this property.)
For such $t_R$ and $t_S$, we took the ratio $R_e$ of the Frobenius norm of the error to the Frobenius norm of the error of the
best rank-$k$ approximation. The resulting points $(k/t_R, R_e-1)$ were generated,
for all test matrices, for three independent trials, resulting in a set of points $P$.

The test matrices are from the University of Florida 
Sparse Matrix Collection,
essentially most of those with at most $10^5$ nonzero entries, and with $n$ up to about  7000.
There were 1155 matrices tested, from 70 sub-collections of matrices, each such
sub-collection representing a particular application area.

The curve in Figure 1 represents the results of these tests, where for a particular point
$(x,y)$ on the curve, at most one percent of points $(t/k_R, R_e-1)\in P$ gave a result where
$k/t_R < x$ but $R_e-1 > y$.

Figure~2 shows a similar curve for the points $(t_R/t_S, \cond(SU)-1)$; thus the necessary
ratio $t_R/t_S$, so that $\cond(SU)\le 1.2$, as for the results in FIgure~1, need be no smaller than
about $1/110$.

%
%
%
%

\begin{figure*}
\begin{center}
  \includegraphics[scale=0.8]{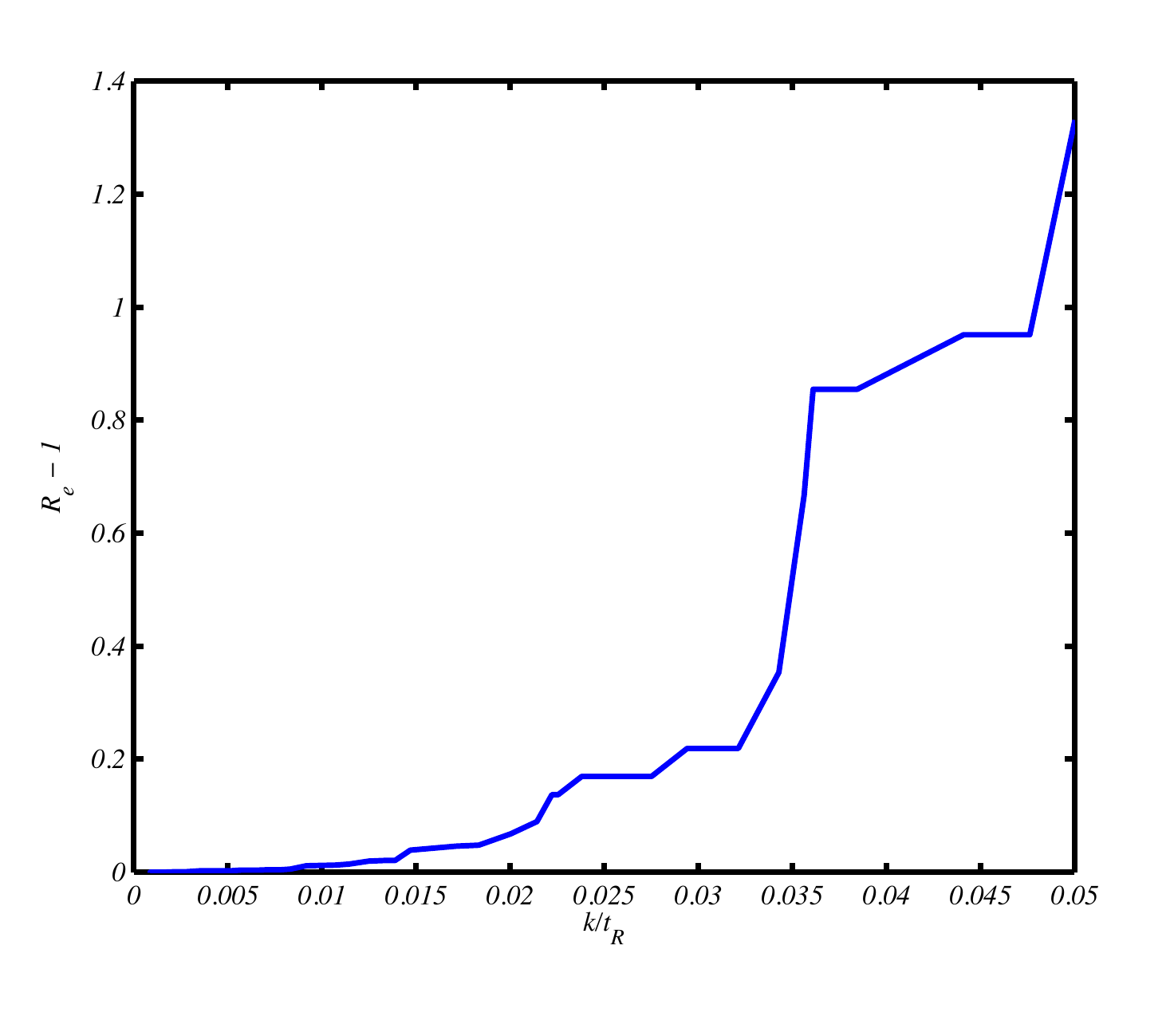} 
 \end{center}
\caption{A ``1\%-Pareto'' curve of error as a function of the size of $\hat R$}
\end{figure*}

\begin{figure*}
\begin{center}
  \includegraphics[scale=0.8]{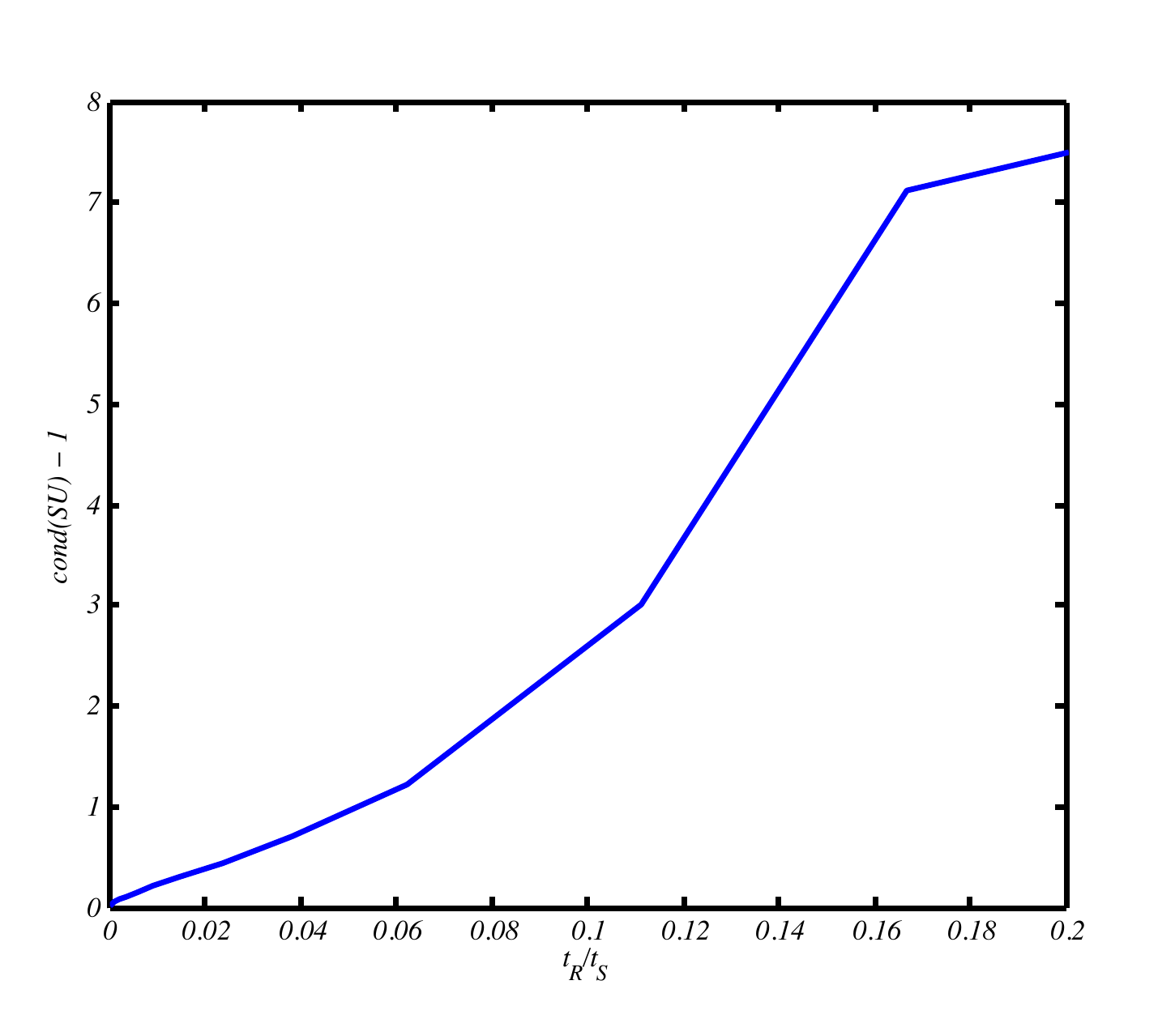} 
 \end{center}
\caption{A 1\%-Pareto curve of $\cond(SU)-1$  as a function of the size of $\hat S$ relative to $\hat R$}
\end{figure*}


\fi

%
%
%

\section*{Acknowledgements}
We acknowledge the support from XDATA program of the Defense Advanced Research Projects Agency (DARPA), administered through Air Force Research Laboratory contract FA8750-12-C-0323.
We thank Jelani Nelson and the anonymous STOC referees for helpful comments.

\bibliographystyle{plain}
\bibliography{main}

\ifSTOC\else
\appendix

\section{Deferred proofs}

\begin{proofof}{of Lemma \ref{lem:betabound}}
Since by assumption $A$ has orthonormal columns,
$\normF{A(\tilde X - X^*)} = \normF{A^\top A(\tilde X - X^*)}$,
so it suffices to bound the latter, or $\normF{\beta}$ where
$\beta\equiv A^\top A(\tilde X - X^*)$.
By Fact \ref{fact:normal}, we have 
\begin{equation}\label{eq:approxlsnormal}
A^\top  \ZZ ^\top  \ZZ  (A \tilde{X} - B) = 0.
\end{equation}
To bound $\normF{\beta}$, we bound
$\normF{A^\top  \ZZ ^\top  \ZZ  A \beta}$, and then show that
this implies that $\normF{\beta}$ is small.
Using that $A A^\top  A = A$ and (\ref{eq:approxlsnormal}), we have
\begin{align*}
A^\top  \ZZ ^\top  \ZZ  A \beta
        & = A^\top  \ZZ ^\top  \ZZ  A A^\top  A (\tilde{X} - X^*)
        \\ & = A^\top  \ZZ ^\top  \ZZ  A (\tilde{X} - X^*)
        \\ & = A^\top  \ZZ ^\top  \ZZ  A(\tilde{X} - X^*) + A^\top  \ZZ ^\top  \ZZ  (B-A\tilde{X})
        \\ & = A^\top  \ZZ ^\top  \ZZ  (B - AX^*).
\end{align*}
Using the hypothesis of the theorem, 
\begin{align*}
\normF{A^\top  \ZZ ^\top  \ZZ  A \beta}
        = \normF{A^\top  \ZZ ^\top  \ZZ  (B - A X^*)}
        \le \sqrt{\epsilon/r}\normF{A}\normF{B-A X^*}
        \le \sqrt{\epsilon}\normF{B - A X^*}.
\end{align*}
To show that this bound implies that $\normF{\beta}$
is small, we use the subadditivity of $\normF{}$ and
the property of any conforming matrices $C$ and $D$,
that $\normF{CD}\le \norm{C}_2\normF{D}$,
to obtain
\begin{align*}
\normF{\beta}
       \le \normF{A^\top  \ZZ ^\top  \ZZ  A \beta} +  \normF{A^\top  \ZZ ^\top  \ZZ  A \beta- \beta}
       \le \sqrt{\epsilon}\normF{B - A X^*} + \norm{A^\top  \ZZ ^\top  \ZZ  A - I }_2 \normF{\beta}.
\end{align*}
By hypothesis,
$\norm{\ZZ A x}^2 = (1\pm\epsilon_0)\norm{x}^2$ for all $x$,
so that $A^\top  \ZZ ^\top  \ZZ  A - I$ has eigenvalues bounded in magnitude by $\epsilon_0^2$,
which implies singular values with the same bound, so that 
$\norm{A^\top  \ZZ ^\top  \ZZ  A - I }_2\le \epsilon^2_0$.
Thus
$\normF{\beta}\le \sqrt{\epsilon}\normF{B - AX^*} + \epsilon^2_0\normF{\beta}$,
or
\begin{equation*}
\normF{\beta} \le \sqrt{\epsilon}\normF{B - A X^*} / (1-\epsilon^2_0) \le 2\sqrt{\epsilon}\normF{B - A X^*},
\end{equation*}
since $\epsilon^2_0\le 1/2$. This bounds $\normF{\beta}$,
and so proves the lemma.
\end{proofof}

\begin{proofof}{of Lemma~\ref{lem:len fixed sparse}}
Let $S = \Phi D$ with associated hash function $h:[n] \rightarrow [t]$. 
For $A_i$ denoting the $i$-th column of $A$, 
let $A_i(b)$ denote the column vector whose $\ell$-th coordinate is $0$ if 
$h(\ell) \neq b$, and whose $\ell$-th coordinate is $A_{\ell,i}$ if $h(\ell) = b$.
We use the second moment method to bound $\|SA\|_F^2$. For the expectation,
\begin{eqnarray}\label{eqn:exp}
{\bf E}_{D,h}[\|SA\|_F^2] = \sum_{i \in [d]} {\bf E}_{D,h}[\|SA_i\|_2^2]
= \sum_{i \in [d]} \sum_{b \in [t]} {\bf E}_{D,h}[(\sum_{\ell \mid h(\ell) = b} A_{\ell, i} D_{\ell, \ell} )^2]
= {\bf E}_h \left [\sum_{i \in [d]} \sum_{b \in [t]} \|A_i(b)\|_2^2 \right ]
= \|A\|_F^2.
\end{eqnarray}
For the second moment, 
\begin{eqnarray}\label{eqn:var}
{\bf E}_{D,h}[\|SA\|_F^4] = \sum_{i \in [d]} {\bf E}_{D,h}[\|SA_i\|_2^4]
+ \sum_{i \neq j \in [d]} {\bf E}_{D,h}[\|SA_i\|_2^2 \cdot \|SA_j\|_2^2].
\end{eqnarray}
We handle the first term in (\ref{eqn:var}) as follows:
\begin{eqnarray*}
{\bf E}_{D,h} [\|SA_i\|_2^4] 
& = & {\bf E}_h \left [\sum_{b,b' \in [t]} {\bf E}_D [(SA_i)_b^2 \cdot (SA_i)_{b'}^2] \right ]\\
& = & {\bf E}_h \left [\sum_{b \in [t]} {\bf E}_D [(SA_i)_b^4] + \sum_{b \neq b' \in [t]} {\bf E}_D [(SA_i)_b^2] \cdot {\bf E}_D [(SA_i)_{b'}^2] \right ]\\
& = & {\bf E}_h \left [\sum_{b \in [t]} {\bf E}_D [(\sum_{\ell \mid h(\ell) = b} A_{\ell, i} D_{\ell, \ell})^4]
+ \sum_{b \neq b' \in [t]} {\bf E}_D [(\sum_{\ell \mid h(\ell) = b} A_{\ell, i} D_{\ell, \ell})^2] \cdot
{\bf E}_D [(\sum_{\ell \mid h(\ell) = b'} A_{\ell, i} D_{\ell, \ell})^2]\right ]\\
& \leq & {\bf E}_h \left [\sum_{b \in [t]} \left (\sum_{\ell \mid h(\ell) = b} A^4_{\ell, i}  +
{4 \choose 2} \sum_{\ell < \ell' \mid h(\ell) = h(\ell') = b} A^2_{\ell, i} A^2_{\ell',i} \right )
+ \sum_{b \neq b' \in [t]} \|A_i(b)\|_2^2 \cdot \|A_i(b')\|_2^2 \right ]\\
& \leq & {\bf E}_h \left [\|A_i\|_4^4 \right ]+ \frac{6}{t} \|A_i\|_2^4 + {\bf E}_h \left [\sum_{b \neq b' \in [t]} \|A_i(b)\|_2^2 \cdot \|A_i(b')\|_2^2 \right ]\\
& \leq & {\bf E}_h \left [\sum_{b \in [t]} \|A_i(b)\|_2^4 \right ]+ \frac{6}{t} \|A_i\|_2^4 + 
{\bf E}_h \left [\sum_{b \neq b' \in [t]} \|A_i(b)\|_2^2 \cdot \|A_i(b')\|_2^2 \right ]\\
& \leq & \frac{6}{t} \|A_i\|_2^4 + \|A_i\|_2^4.
\end{eqnarray*}
For the second term in (\ref{eqn:var}), for $i \neq j \in [d]$,
\begin{eqnarray*}
{\bf E}_{D,h}[\|SA_i\|_2^2 \cdot \|SA_j\|_2^2]
& = & {\bf E}_{D,h} \left [\sum_{b \in [t]} \left (\sum_{\ell \mid h(\ell) = b} A_{\ell, i} D_{\ell, \ell} \right )^2
\left (\sum_{\ell' \mid h(\ell') = b} A_{\ell', j} D_{\ell', \ell'} \right )^2 \right ]\\
&& + {\bf E}_{D,h} \left [\sum_{b \neq b' \in [t]} \left (\sum_{\ell \mid h(\ell) = b} A_{\ell, i} D_{\ell, \ell} \right )^2
\left (\sum_{\ell' \mid h(\ell') = b'} A_{\ell', j} D_{\ell', \ell'} \right )^2 \right ]\\
& = & {\bf E}_h \left [\sum_{b \in [t]} 
\left (\sum_{\ell, \ell' \mid h(\ell) = h(\ell') = b} A_{\ell, i} A_{\ell',i} D_{\ell, \ell} D_{\ell', \ell'} \right )
\left (\sum_{\ell, \ell' \mid h(\ell) = h(\ell') = b} A_{\ell, j} A_{\ell',j} D_{\ell, \ell} D_{\ell', \ell'} \right ) \right ]\\
&& + {\bf E}_h \left [\sum_{b \in [t]} \|A_i(b)\|_2^2 \cdot \|A_j(b)\|_2^2 
+ \sum_{b \neq b' \in [t]} \|A_i(b)\|_2^2 \cdot \|A_j(b')\|_2^2 \right ]\\
& = & \|A_i\|_2^2 \cdot \|A_j\|_2^2 + {\bf E}_h \left [\sum_{b \in [t]} 4 \sum_{\ell < \ell' \mid h(\ell) = h(\ell') = b} 
A_{\ell, i} A_{\ell', i} A_{\ell, j} A_{\ell', j} \right ],
\end{eqnarray*}
where the constant $4$ arises because if we choose indices $\ell < \ell'$ from
$\left (\sum_{\ell, \ell' \mid h(\ell) = h(\ell') = b} A_{\ell, i} A_{\ell',i} D_{\ell, \ell} D_{\ell', \ell'} \right )$
we need to choose the same $\ell$ and $\ell'$ from
$\left (\sum_{\ell, \ell' \mid h(\ell) = h(\ell') = b} A_{\ell, j} A_{\ell',j} D_{\ell, \ell} D_{\ell', \ell'} \right )$
in order to have a non-zero expectation, and there are $4$ ways of doing this for distinct $\ell, \ell'$.
Continuing, 
\begin{eqnarray*}
\|A_i\|_2^2 \cdot \|A_j\|_2^2 + {\bf E}_h \left [\sum_{b \in [t]} 4 \sum_{\ell < \ell' \mid h(\ell) = h(\ell') = b} 
A_{\ell, i} A_{\ell', i} A_{\ell, j} A_{\ell', j} \right ]
& \leq & \|A_i\|_2^2 \cdot \|A_j\|_2^2 + 
{\bf E}_h \left [4\sum_{b \in [t]}\langle A_i(b), A_j(b) \rangle^2 \right ]\\
& \leq &  \|A_i\|_2^2 \cdot \|A_j\|_2^2 + 
{\bf E}_h \left [4\sum_{b \in [t]} \|A_i(b)\|_2^2 \cdot \|A_j(b')\|_2^2 \right ]\\
& = & \|A_i\|_2^2 \cdot \|A_j\|_2^2 + 
\frac{4}{t} \sum_{\ell, \ell' \in [n]} A_{\ell, i}^2 A_{\ell, j}^2\\
& = & \left(1 + \frac{4}{t} \right )\|A_i\|_2^2 \cdot \|A_j\|_2^2.
\end{eqnarray*}
Combining (\ref{eqn:exp}) with (\ref{eqn:var}) and the bounds on the terms in (\ref{eqn:var}) above,
\begin{eqnarray*}
{\bf Var}[\|SA\|_F^2]
& \leq & \left (\sum_{i \in [d]} \frac{6}{t} \|A_i\|_2^4 + \|A_i\|_2^4 \right ) 
+ \sum_{i \neq j \in [d]} \left(1 + \frac{4}{t} \right )\|A_i\|_2^2 \cdot \|A_j\|_2^2 - \|A\|_F^2\\
& \leq & \frac{6}{t} \|A\|_F^2\\
& = & \frac{6}{t} {\bf E}[\|SA\|_F^2].
\end{eqnarray*}
The lemma now follows by Chebyshev's inequality, for appropriate $t = \Omega(\eps^{-2})$. 
\end{proofof}

\begin{proofof}{of Lemma~\ref{lem:len fixed srht}}
Lemma 15 of \cite{BG} shows that $\norm{SA} \le (1+\eps)\norm{A}$
with arbitrarily low failure probability,
and the other direction follows from a similar argument. Briefly:
the expectation of $\norm{SA}^2$ is $\norm{A}^2$, by construction,
and Lemma 11 of \cite{BG} implies that with arbitrarily small
failure probability, all rows of $SA$ will have squared norm
at most $\beta \equiv \frac{\alpha}{t}\norm{A}^2$, where $\alpha$ is a value
in $O(\log n)$. Assuming that this bound holds,
it follows from Hoeffding's inequality
that the probability that $| \norm{SA}^2 - \norm{A}^2 | \ge \eps \norm{A}^2$
is at most $2\exp(-2[\eps\norm{A}^2]^2 / t\beta^2)$, or
$2\exp(-2\eps^2t/\alpha^2)$, so that $t = \Theta(\eps^{-2}(\log n)^2 )$ suffices
to make the failure probability at most $1/10$.
\end{proofof}

\begin{proofof}{of Lemma \ref{lem:lpl2}}
This is almost exactly the same as in \cite{CDMMMW}, we simply adjust notation and parameters. 
Applying Theorem \ref{thm:jlmain}, we have that
with probability at least $1-1/(100w)$, for all $x \in \mathbb{R}^r$, if we consider $y = Ax$ and
write $y^T = [z_1^T, z_2^T, \ldots, z_{n/w}^T]$, then for all $i \in [n/w]$,  
\begin{eqnarray*}
\sqrt{\textstyle\frac12}\norm{z_i}_2\le\norm{Sz_i}_2\le
\sqrt{\textstyle\frac32}\norm{z_i}_2
\end{eqnarray*}
By relating the $2$-norm and the $p$-norm, for $1 \leq p \leq 2$, we have
$$\norm{Sz_i}_p
	\le t^{1/p-1/2}\norm{Sz_i}_2
	\le t^{1/p-1/2} \sqrt{\textstyle\frac32} \norm{z_i}_2
	\le t^{1/p-1/2} \sqrt{\textstyle\frac32} \norm{z_i}_p,$$
and similarly,
$$\norm{Sz_i}_p
	\ge \norm{Sz_i}_2
	\ge \sqrt{\textstyle\frac12}\norm{z_i}_2
	\ge \sqrt{\textstyle\frac12} w^{1/2-1/p}\norm{z_i}_p.
$$
If $p > 2$, then 
$$
\norm{Sz_i}_p
	\le \norm{Sz_i}_2
	\le \sqrt{\textstyle\frac32} \norm{z_i}_2
	\le \sqrt{\textstyle\frac32} w^{1/2 - 1/p} \norm{z_i}_p,$$
and similarly,
$$
\norm{Sz_i}_p
	\ge t^{1/p-1/2}\norm{Sz_i}_2
	\ge t^{1/p-1/2} \sqrt{\textstyle\frac12} \norm{z_i}_2
 	\ge t^{1/p-1/2} \sqrt{\textstyle\frac12}  \norm{z_i}_p.
$$
Since $\norm{Ax}_p^p = \norm{y}_p^p = \sum_i \norm{z_i}^p$
and $\norm{FAx}_p^p = \sum_i \norm{Sz_i}_p^p$,
for $p\in [1,2]$ we have with probability $1-1/(100w)$
\[
\sqrt{\textstyle\frac12} w^{1/2-1/p}\norm{Ax}_p
	\le \norm{FAx}_p \le \sqrt{\textstyle\frac32} t^{1/p-1/2} \norm{Ax}_p,
\]
and for $p\in [2,\infty)$ with probability $1-1/(100w)$
\[
\sqrt{\textstyle\frac12} t^{1/p-1/2} \norm{Ax}_p
	\le \norm{FAx}_p \le \sqrt{\textstyle\frac32} w^{1/2 - 1/p} \norm{Ax}_p.
\]
In either case,
\begin{equation}\label{eqn:first}
\norm{Ax}_p \le \gamma_p \norm{FAx}_p \le \sqrt{3} (tw)^{|1/p-1/2|}\norm{Ax}_p.
\end{equation}


Applying Theorem \ref{thm:basisOld}, we have, from the definition
of a $(\alpha,\beta,p)$-well-conditioned basis, that 
\begin{eqnarray}\label{eqn:second}
\|FA U \|_p \leq \alpha 
\end{eqnarray}
and for all $x \in \mathbb{R}^d$,
\begin{eqnarray}\label{eqn:third}
\|x\|_q \leq \beta \|FAU\|_p.
\end{eqnarray}

Combining (\ref{eqn:first}) and (\ref{eqn:second}), we have that with probability at least $1-1/(100w)$,
\begin{eqnarray*}
\|A U/(r\gamma_p) \|_p \leq \sum_i \|A U_i/r\gamma_p \|_p
	\leq \sum_i  \|F A U_i/r\|_p
	\leq \alpha.
\end{eqnarray*}
Combining (\ref{eqn:first}) and (\ref{eqn:third}), we have that with probability at least $1-1/(100w)$,
for all $x \in \mathbb{R}^r$,
\begin{eqnarray*}
\|x\|_q \leq \beta \|F A U x\|_p
	\leq \beta \sqrt{3} r (tw)^{|1/p-1/2|} \|A U\frac{1}{r\gamma_p}  x\|_p.
\end{eqnarray*}
Hence $A U/(r\gamma_p)$ is an 
$(\alpha, \beta \sqrt{3} r (tw)^{|1/p-1/2|} , p)$-well-conditioned
basis. The time to compute
$F A$ is $O(\nnz(A) \log n)$ by Theorem \ref{thm:jlmain}. Notice
that $FA$ is an $nt/w \times n$ matrix, which is $O(n/r^5) \times r$, 
and so the time to 
compute $U$ from $F A$ is $O((n/r^5) r^5 \log n) = O(\nnz(A) \log n)$, 
since $\nnz(A) \geq n$.  
\end{proofof}
\fi

\end{document}